\documentclass[journal]{IEEEtran}
\usepackage{lipsum, babel}
\usepackage{multicol}

\usepackage{setspace}
\usepackage{lettrine}

\usepackage{graphics,
	psfrag,
	epsfig,
	amsthm,
	cite,
	amssymb,
	url,
	dsfont,
	algorithmic,
	balance,
	enumerate,
	color,
	setspace
}
\usepackage{amsmath}

\usepackage{epstopdf}

\newtheorem{lemma}{Lemma}

\newtheorem{theorem}{Theorem}

\newcommand{\eref}[1]{(\ref{#1})}
\newcommand{\sref}[1]{Section~\ref{#1}}

\newcommand{\cref}[1]{Constraint~\ref{#1}}
\newcommand{\thref}[1]{Theorem~\ref{#1}}

\newcommand{\lref}[1]{Lemma~\ref{#1}}

\newcommand{\ignore}[1]{}

\usepackage[labelformat=simple]{subcaption}

\IEEEoverridecommandlockouts
\usepackage{mathtools}
\usepackage{color}
\usepackage[utf8]{inputenc}
\usepackage[ruled, lined, longend]{algorithm2e}
\SetEndCharOfAlgoLine{}

\begin{document}
	
	\title{
	Decentralized Model Dissemination Empowered Federated Learning in mmWave Aerial-Terrestrial Integrated Networks}
	
	\setlength{\columnsep}{0.21 in}

	\author{
		\IEEEauthorblockN{Mohammed S. Al-Abiad, \textit{Member, IEEE}, Md. Zoheb Hassan, and Md. Jahangir Hossain, \textit{Senior Member, IEEE}} 
		
	\thanks {
		Mohammed S. Al-Abiad is with the Department of Electrical and Computer Engineering, University of Toronto, Toronto, ON M5S, Canada (e-mail: mohammed.saif@utoronto.ca).
  
  Md. Zoheb Hassan is with Wireless@Virginia Tech, Bradley Department of ECE,  Virginia Tech, VA, USA (e-mail:  mdzoheb@vt.edu).
  
  Md. Jahangir Hossain is with the School of
		Engineering, University of British Columbia, Kelowna, BC V1V 1V7, Canada
		(e-mail: jahangir.hossain@ubc.ca). 
	}
		\vspace{-0.99cm}
	}
	
	\maketitle
\vspace{-0.99cm}
	\begin{abstract}
	
It is anticipated that aerial-terrestrial integrated networks incorporating unmanned aerial vehicles (UAVs) mounted relays will offer improved coverage and connectivity in the beyond 5G era. Meanwhile, federated learning (FL) is a promising distributed machine learning technique for building inference models over wireless networks due to its ability to maintain user privacy and reduce communication overhead. However, off-the-shelf FL models aggregate global parameters at a central parameter server (CPS), increasing energy consumption and latency, as well as inefficiently utilizing radio resource blocks (RRBs) for distributed user devices (UDs). This paper presents a resource-efficient FL framework, called FedMoD (\textbf{fed}erated learning with \textbf{mo}del \textbf{d}issemination), for millimeter-wave (mmWave) aerial-terrestrial integrated networks with the following two unique characteristics. Firstly, FedMoD  presents a novel decentralized model dissemination algorithm that makes use of UAVs as local model aggregators through UAV-to-UAV and device-to-device (D2D) communications. As a result, FedMoD (i) increases the number of participant UDs in developing FL model and (ii) achieves global model aggregation without involving CPS. Secondly, FedMoD reduces the energy consumption of FL using radio resource management (RRM) under the constraints of over-the-air learning latency. In order to achieve this, by leveraging graph theory, FedMoD optimizes the scheduling of line-of-sight (LOS) UDs to suitable UAVs/RRBs over mmWave links and  non-LOS UDs to available LOS UDs via overlay D2D communications. Extensive simulations reveal that decentralized FedMoD offers same convergence rate performance as compared to  conventional FL frameworks.

\end{abstract}

\begin{IEEEkeywords}
Decentralized FL model dissemination, energy consumption, UAV communications. 
\end{IEEEkeywords}

\vspace{-0.44cm}

\section{Introduction} 
Unmanned aerial vehicles (UAVs) are expected to have a significant impact on the economy by 2026 with a projected global market value of US\$59.2 billion, making the inclusion of UAVs critical in beyond 5G cellular networks \cite{UAV_economy}. There are several unique features of UAV-mounted communication platforms, including the high likelihood of establishing line-of-sight connections with ground nodes, rapid deployment, and adjustable mobility \cite{Zoheb_IoD}. With such attributes, UAVs can serve as aerial base stations (BSs) or relays in conjunction with terrestrial base stations, resulting in aerial-terrestrial integrated networks (ATINs). By connecting cell-edge user devices (UDs) to terrestrial cellular networks via aerial BSs or relays, ATINs improve coverage and connectivity significantly \cite{Rui_Zhang_UAV}. 
The 3GPP standard also incorporates the use of UAVs as a communication infrastructure to complement terrestrial cellular networks \cite{3GPP}. During current 5G deployment efforts, it has been shown that the millimeter-wave band at 28 GHz is significantly larger and more capable than the sub-6 GHz band. At the same time, air-to-ground communications have the advantage of avoiding blockages and maintaining LOS connectivity as a result of UAV's high altitude and flexibility. \cite{UAV_mmWave}. Therefore, the mmWave band is suitable for deploying high-capacity ATINs in the next-generation cellular networks.

A data-driven decision making process enables wireless networks to manage radio resources more efficiently by predicting and analyzing several dynamic factors, such as users' behavior, mobility patterns, traffic congestion, and quality-of-service expectations. Data-driven radio resource management (RRM) has gained increasing popularity, thanks to the expansion of wireless sensing applications, the availability of enormous data, and the increasing computing capabilities of devices. To train machine learning (ML) models, raw data collected from individual UDs is aggregated in a central parameter server (CPS). As a result, such centralized ML approaches require enormous amounts of network resources to collect raw data from UDs. In addition, centralized ML also impairs users' privacy since CPS can easily extract sensitive information from raw data gathered from UDs.  Recently, Google proposed federated learning (FL) for UDs to collaboratively  learn a model without sharing their private data \cite{FL_Google}. In FL, UDs update parameters according to their local datasets, and only the most recent parameters are shared with the CPS. Using local models from all participating UDs, the CPS updates global model parameters and shares them with the UDs. The local and global models are adjusted iteratively until convergence. Unlike centralized ML approaches, FL not only protects UD privacy but also improves wireless resource utilization significantly. Nevertheless, the convergence performance of FL in wireless networks significantly depends on the appropriate selection of the participating UDs, based on both channel and data quality, and bandwidth allocation among the selected UDs \cite{FL_Reed}.

The FL framework provides a powerful computation tool for ATINs to make decentralized decisions \cite{UAV_FL_1}. UAVs are frequently used as aerial sensors or aerial data collectors in several practical scenarios, and the convergence and accuracy of FL in such use cases can be improved by appropriately exploiting the  unique attributes of air-to-ground communication links. For instance, a FL framework was developed for hazardous zone detection and air quality prediction by utilizing UAVs as local learners and deploying a swarm of them to collect local air quality index (AQI) data \cite{UAV_FL_2}.  A UAV-supported FL framework was also proposed in which a drone visited learning sites sequentially, aggregated model parameters locally, and relayed them to the CPS for global aggregation \cite{UAV_FL_2}. Meanwhile, ATINs can also deploy UAVs as aerial BSs with edge computing capabilities. In this context, UAV can provide global model aggregation capability for a large numbers of ground UDs, thanks to its large coverage and high probability of establishing LOS communications \cite{UAV_FL_3}. In the aforesaid works, all the local model parameters were aggregated into a single CPS using the conventional star-based FL framework. Although such a star-based FL is convenient, it poses several challenges in the context of ATINs. Firstly, a star-based FL requires a longer convergence time due to the presence of straggling local learners. Recall, the duration of transmission and hovering of a UAV influences its energy consumption, and consequently, increasing the convergence time of FL directly increases the energy consumption of UAVs. This presents a significant challenge for implementing FL in ATINs since UAVs usually have limited battery capacity.  In addition, as a result of increased distance and other channel impairments, a number of local learners with excellent datasets may be out of coverage of the central server in practice. The overall learning accuracy of FL models can be severely impacted if these local learners are excluded from the model. The use of star-based FL frameworks in ATINs is also confronted by the uncertainty of air-to-ground communication links resulting from random blocking and the mobility of UAVs. This work seeks to address these challenges by proposing a resource-efficient FL framework for mmWave ATINs that incorporates decentralized model dissemination and energy-efficient UD scheduling.
\vspace{-0.55cm}
 \subsection{Summary of the Related Works}
 \vspace{-0.25cm}
In the current literature, communication-efficient FL design problems are explored. In \cite{FL_Lit_1_1}, the authors suggested a stochastic alternating direction multiplier method to update the local model parameters while reducing communications between local learners and CPS. 
In \cite{FL_Lit_1_3}, a joint client scheduling and RRB allocation scheme was developed to minimize accuracy loss. To minimize the loss function of FL training, UD selection, RRB scheduling, and transmit power allocation were optimized simultaneously \cite{FL_Lit_1_4}. Numbers of global iterations and duration of each global iteration were minimized by jointly optimizing the UD selection and RRB allocation \cite{FL_Lit_1_5}. Besides, since UDs participating in FL are energy-constrained, an increasing number of studies focused on designing energy-efficient FL frameworks. As demonstrated in \cite{FL_Lit_1_6}, the energy consumption of FL can be reduced by uploading only quantized or compressed model parameters from UDs to CPS. Furthermore, RRM enhances energy efficiency of FL in large-scale networks. Several aspects of RRM, such as client scheduling, RRB allocation, and transmit power control, were extensively studied to minimize both communication and computation energy of FL frameworks \cite{FL_Lit_1_7, FL_Lit_1_8}. An energy-efficient FL framework based on relay-assisted two-hop transmission and non-orthogonal multiple access scheme was recently proposed for both energy and resource constrained Internet of Things (IoT) networks \cite{FL_Lit_1_10}. In the aforesaid studies, conventional star-based FL frameworks were studied. Due to its requirement to aggregate all local model parameters on a single server, the  star-based FL is inefficient for energy- and resource-constrained wireless networks.
 
 Hierarchical FL (HFL) frameworks involve network edge devices uploading model parameters to mobile edge computing (MEC) servers for local aggregation, and MEC servers uploading aggregated local model parameters to CPS periodically. The HFL framework increases the number of connected UDs and reduces energy consumption \cite{HFL_1_1}. To facilitate the HFL framework, a client-edge-cloud collaboration framework was explored \cite{HFL_1_2}. HFL was investigated in heterogeneous wireless networks through the introduction of fog access points and multiple-layer model aggregation \cite{HFL_1_3}. Dynamic wireless channels in the UD-to-MEC and MEC-to-CPS hops and data distribution play a crucial role in FL learning accuracy and convergence. Thus, efficient RRM is imperative for implementation of HFL. As a result, existing literature evaluated several RRM tasks, including UD association, RRB allocation, and edge association, to reduce cost, latency, and learning error of HFL schemes \cite{HFL_1_4, HFL_1_5}.
  
 While HFL increases the number of participating UDs, its latency and energy consumption are still hindered by dual-hop communication for uploading and broadcasting local and global model parameters. Server-less FL is a promising alternative to reduce latency and energy consumption. This FL framework allows UDs to communicate locally aggregated models without involving central servers, thereby achieving model consensus. The authors in \cite{FL_Dec_1_1} proposed a FL scheme that relies on device-to-device (D2D) communications to achieve model consensus. However, due to the requirement of global model aggregation with two-time scale FL over both D2D and user-to-CPS wireless transmission, this FL scheme has limited latency improvement. In \cite{FL_Dec_1_2, FL_Dec_1_3}, the authors developed FL model dissemination schemes by leveraging connected edge servers (ESs), which aggregate local models from their UD clusters and exchange them with all the other ESs in the network for global aggregation. However, a fully connected ES network is prohibitively expensive in practice, especially when ESs are connected by wireless links. In addition, each global iteration of FL framework takes significantly longer because ESs continue to transmit local aggregated models until all other ESs receive them successfully \cite{FL_Dec_1_2, FL_Dec_1_3}. The authors in \cite{FL_Dec_1_4} addressed this issue by introducing conflicting UDs, which are the UDs covering multiple clusters, and allowing parameter exchanges between them and local model aggregators.
 
In spite of recent advances in resource-efficient, hierarchical, and decentralized FL frameworks, existing studies have several limitations in utilizing UAVs as local model aggregators in mmWave ATINs. In particular, state-of-the-art HFL schemes of \cite{HFL_1_2, HFL_1_3, HFL_1_4} can prohibitively increase the communication and propulsion energy consumption of UAVs because it involves two-hop communications and increased latency. Additionally, mmWave band requires LOS links between UDs and UAVs for local model aggregation, as well as LOS UAV-to-UAV links for model dissemination. Accordingly, the FL model dissemination frameworks proposed in \cite{FL_Dec_1_2, FL_Dec_1_3, FL_Dec_1_4} will not be applicable to mmWave ATINs. We emphasize that in order to maintain convergence speed and reduce energy consumption, the interaction among UD-to-UAV associations, RRB scheduling, and UAV-to-UAV link selection, in addition to the inherent properties of mmWave bands, must be appropriately characterized. Such a fact motivates us to develop computationally efficient models dissemination and RRM schemes for mmWave ATINs implementing decentralized FL.

\vspace{-0.44cm}
\subsection{Contributions}
This work proposes a resource-efficient and fast-convergent FL framework for mmWave ATINs, referred to as Federated Learning with Model Dissemination (FedMoD). 
The specific contributions of this work are summarized as follows.

\begin{itemize}
    \item A UAV-based distributed FL model aggregation method is proposed by leveraging UAV-to-UAV communications. Through the proposed method, each UAV is able to collect local model parameters only from the UDs in its coverage area and share those parameters over LOS mmWave links with its neighbor UAVs. The notion of physical layer network coding is primarily used for disseminating model parameters among UAVs. This allows each UAV to collect all of the model parameters as well as aggregate them globally without the involvement of the CPS. With the potential to place UAVs near cell edge UDs, the proposed UAV-based model parameter collection and aggregation significantly increases the number of participating UDs in the FL model construction process. Based on the channel capacity of the UAV-to-UAV links, a conflict graph is formed to facilitate distributed model dissemination among the UAVs and a maximal weighted independent search (MWIS) method is proposed to solve the conflict graph problem. In light of the derived solutions, a decentralized FedMoD is developed and its convergence is rigorously proved.
    \item Additionally, a novel RRM scheme is investigated to reduce the overall energy consumption of the developed decentralized FL framework under the constraint of learning latency. The proposed RRM optimizes both (i) the scheduling of LOS UDs to suitable UAVs and radio resource blocks (RRBs) over mmWave links and (ii) the scheduling of non-LOS UDs to LOS UDs over side-link D2D communications such that non-LOS can transmit their model parameters to UAVs with the help of available LOS UDs. As both scheduling problems are provably NP-hard, their optimal solutions require prohibitively complex computational resources. Two graph theory solutions are therefore proposed for the aforementioned scheduling problems to strike a suitable balance between optimality and computational complexity.
    
    \item To verify FedMoD's effectiveness over  contemporary star-based FL and  HFL schemes, extensive numerical simulations are conducted. Simulation results reveal that FedMoD achieves good convergence rates and superior energy consumption compared to the benchmark schemes.
    
\end{itemize}

The rest of this paper is organized as follows. In Section II, the system model described in detail. In Section,  the proposed FedMoD algorithm is explained thoroughly along with its convergence analysis. Section IV presents the RRM scheme for improving energy-efficiency of the proposed FedMoD framework. Section V presents various simulation results on the performance of the proposed FedMoD scheme. Finally, the concluding remarks are provided in Section VI.

\vspace{-0.4cm}
\section{System Model} \label{SMMM}

\subsection{System Overview}

The envisioned mmWave aerial-terrestrial integrated network (ATIN) model is illustrated in Fig. \ref{fig1}, that consists of  a single CPS, multiple UAVs that are connected to each other through mmWave air-to-air (A2A) links, and multiple UDs that are under the serving region of each UAV. The UDs are connected with the UAVs via mmWave. The set of all the considered UDs is denoted by $\mathcal U=\{1,2, \cdots, U\}$ and the set of UAVs is denoted by   $\mathcal K=\{1, 2,  \cdots, K\}$. The federated learning (FL) process is organized in
iterations, indexed by $\mathcal T=\{1, 2,   \cdots, T\}$. Similar to \cite{RRB1, AhmedH}, each UAV $k$ has a set of orthogonal RRBs that is denoted by $\mathcal B=\{1, 2,  \cdots, B\}$, and the UDs are scheduled to these RRBs to offload their local parameters to the UAVs. The
set of UDs in the serving region of the $k$-th UAV is denoted by $\mathcal U_k=\{1, 2, \cdots, U_k\}$. In addition, for the $u$-th UD, the available UAVs are denoted by a set $\mathcal K_u$.  Therefore, some UDs are able to access multiple UAVs simultaneously. We assume that  (i) the $u$-th UD is only associated to the $k$-th UAV during the $t$-th FL iteration and (ii) neighboring UAVs transmit FL models to the scheduled UDs over orthogonal RRBs. In this work, we offload the CPS for performing global aggregations. However, the CPS is required to coordinate the clustering optimization of UAVs and their associated UDs through reliable control channels.

Suppose that UAV $k$ flies and hovers  at a fixed flying altitude $H_k$, and it is assumed that all the UAVs have the same altitude. Let $\mathbf x_l=(x_k, y_k, H_k)$ is the 3D location of the $k$-th UAV
and $(x_u, y_u)$ is the 2D location of the $u$-th UD. In accordance with \cite{elev}, for the mmWave UD-UAV communications to be successful, one needs to ensure LOS connectivity between UAVs and UDs. However, some of the UDs may not have LOS communications to the UAVs, thus they can not transmit their trained local parameters directly to the UAVs. Let $\mathcal U_{los}$ be the set of UDs who have LOS links to the UAVs, and let $\mathcal U_{non}$ be the set of UDs who do not have LOS links to the UAVs. Given an access link between the $u$-th UD, i.e., $u\in \mathcal U_{los}$, and the $k$-th UAV, the path loss of the channel (in dB) between  the $u$-th UD and the $k$-th UAV is expressed as follows $
PL(u,k)=20\log_{10}(\frac{4 \pi f_c d_{u,k}}{c})$, where
$f_c$ is the carrier frequency, and $c$ is the light speed, and $d_{u,k}$ is the distance between the $u$-th UD and the $k$-th UAV \cite{elev}. The wireless channel gain between the $u$-th UD and the $k$-th UAV on the $b$-th RRB is $
h^u_{k,b}=10^{-PL(u,k)/10}$. Let $p$ be the transmission power of the UDs and maintains fixed  and $N_o$ as the AWGN noise power.  Therefore, the achievable capacity at which the $u$-th UD can transmit its local model parameter to the $k$-th UAV on the $b$-th RRB at the $t$-th global iteration is given by Shannon’s formula $
R^u_{k,b}=W\log_{2}(1+\frac{p |h^u_{k,b}|^2}{N_0}), \forall u\in \mathcal U_k, k \in \mathcal K_u$,
where $\mathcal U_k\subset \mathcal U_{los}$ and $W$ is the RRB's bandwidth. Note that the transmission rate between the $u$-th UD and the $k$-th UAV on the $b$-th RRB determines if the $u$-th UD is covered by the $k$-th corresponding UAV and has LOS to the $u$-th UD. In other words,  the $u$-th UD is within the coverage of the $k$-th corresponding UAV if $R^u_{k,b}$ meets the rate threshold $R_0$, i.e., $R^u_{k,b}\geq R_0$, and has LOS link to the $k$-th UAV. Each UAV $k$ aggregates the local models of its scheduled UDs only.

For disseminating the local aggregated models among the UAVs to reach global model consensus, UAVs can communicate  through A2A links. Thus, the A2A links between the UAVs are assumed to be in LOS condition \cite{17}. We also
assume that the UAVs employ directive beamforming to
improve the rate. As a result, the gain of the UAV antenna located at $x_k$, denoted by $G^A$, at the receiving UAV is given by \cite{18}
\begin{equation}
	\begin{split}
G^A(d_{A, x_k})=\begin{cases}
	& G^A_m,  ~\text{if} ~-\frac{\theta^a_b}{2} \leq \Phi \leq \frac{\theta^a_b}{2} \\
	& G^A_s, ~\text{otherwise},
	\end{cases}
	\end{split}
	\end{equation}
where $d_{A, x_k}$ is the distance between the typical receiving UAV and the $k$-th UAV at $x_k$, $G^A_m, G^A_s$ are the gains of the main-lobe and side-lobe, respectively, $\Phi$ is the sector angle, and  $\theta^a_b\in [0, 180]$ is the beamwidth in degrees \cite{19}. Accordingly, the received power at the typical receiving UAV from UAV $k$ at $x_k$ is given by
\begin{equation}
    P^A_{r,k}=PG^A(d_{A, x_k})\zeta_A H_A^{x_k}d^{-\alpha_A}_{A, x_k},
\end{equation}
where $\zeta_A$ represents the excess losses, $H_A^{x_k}$ is the Gamma-distributed channel power gain, i.e., $H_A^{x_k} \approx \Gamma(m_A, \frac{1}{m_A})$, with a fading parameter $m_A$, and $\alpha_A$ is the path-loss exponent. As a result, the SINR at the typical receiving UAV is given by 
\begin{equation}
    \gamma=\frac{\mu_A H_A^{x_k} d^{-\alpha_A}_{A, x_k}}{I+\sigma^2}, 
\end{equation}
where $\mu_A=P_A G^A_m\zeta_A$, $I$ is the interference power. Such interference can be expressed as follows
\begin{equation}
    I= \sum^{K}_{j=1, j\neq k} P G^A(d_{A,x_j})\zeta_AH^{x_j}_Ad^{-\alpha_A}_{A, x_j},
\end{equation}
where $G^A(d_{A,x_j})=G^A_m$ with a probability of $q_A$ and $G^A(d_{A,x_j})=G^A_s$ with a probability of $1-q_A$.

Once the local aggregated model dissemination among the UAVs is completed, the $k$-th UAV adopts a common transmission rate $R_k$ that is equal to the minimum achievable rates of all its scheduled UDs  $\mathcal U_{k}$. This adopted transmission rate is $R_k=\min_{u\in \mathcal U_{k}}R^k_{u}$, which is used to transmit the global model to the UDs to start the next global iteration.


\begin{figure}[t!]
	\centering
	\includegraphics[width=0.75\linewidth]{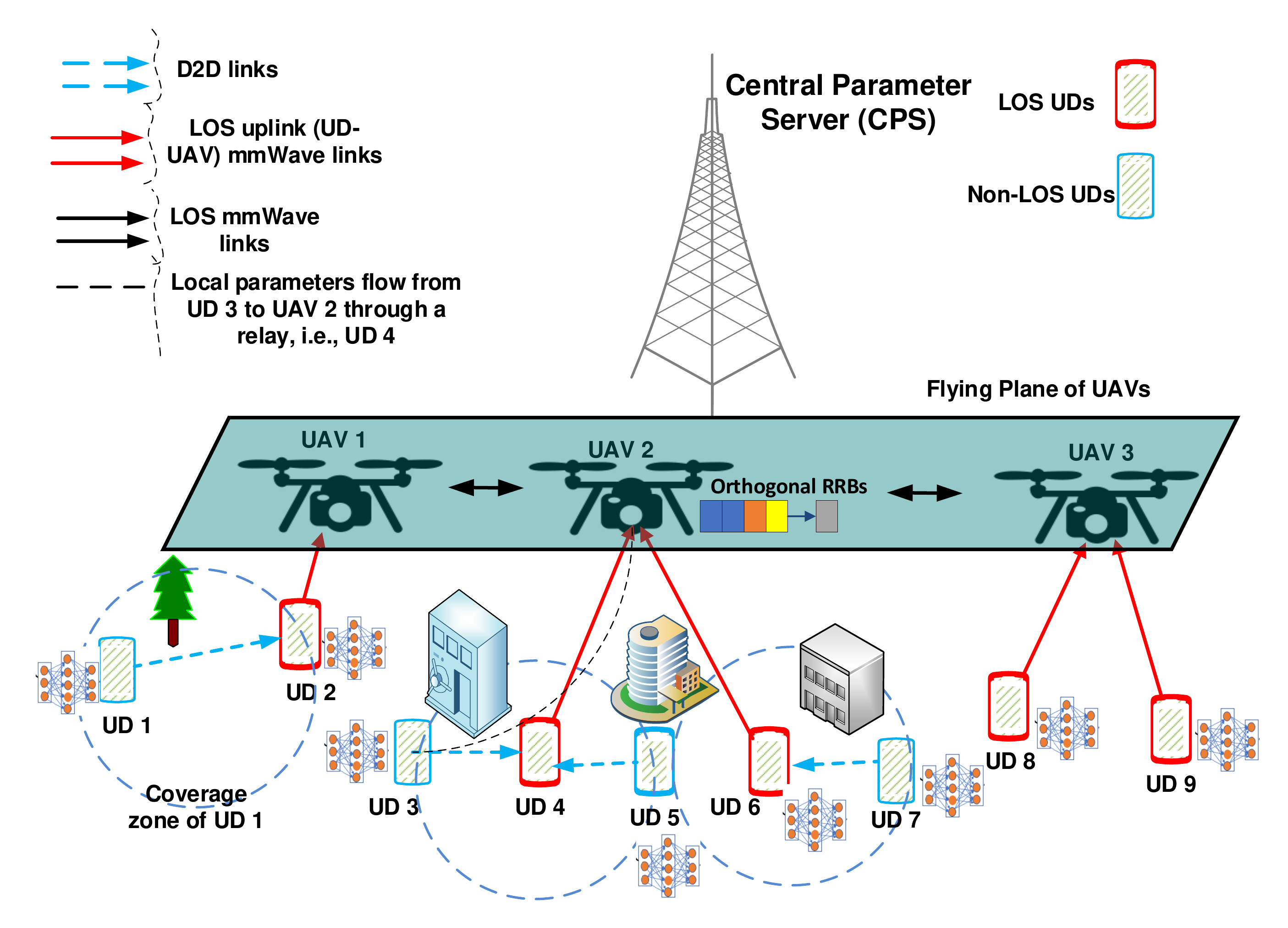}
	\caption{ATIN network with one CPS, $3$
		UAVs, $9$ UDs, and a set of RRBs per each UAV. For instance, UDs $1$, $3$, $5$, and $7$ do not have direct LOS links to the UAVs. Thus, they transmit the trained local models to the UAVs via LOS UDs, e.g.,   UDs $2$, $4$, and $6$. UDs $8$ and $9$ can transmit their models directly to the UAVs via LOS mmWave links.}
	\label{fig1}
\end{figure}

\vspace{-0.55cm}
\subsection{Transmission Time Structure}
The UAVs start local model aggregations  after receiving the local trained models of the scheduled UDs across all the RRBs. Since different UDs $\mathcal U_{los}$ will have different transmission rates, they will have different transmission durations for uploading their trained parameters to the UAVs/RRBs. Let $s$ be the size of the UD’s local vector parameter 
(which is the same for the global model), expressed in bits. Note that the analysis in this subsection is for the transmission duration of one global iteration $t$. For simplicity, we represent $X$ as the number of elements in the set $\mathcal X$. The time required by the $u$-th UD, $u\in \mathcal U_{los}$, to reliably transmit its model update to the $k$-th selected UAV over the $b$-th RRB is then given by $T_u=\frac{s}{R^u_{k,b}}$. With this consideration, we can see that, given the number of participating UDs $U_{los}$, the transmission duration is $\max_{u \in  \mathcal U_{los}}\{T_u\}=\max_{u \in \mathcal U_{los}}\frac{s}{R^u_{k,b}}$. When $U_{los}$ is large, $\max_{u \in  \mathcal U_{los}}\{T_u\}$ can dramatically grow. The transmission duration is therefore constrained by the minimum rate of the scheduled UDs $\mathcal U_{los}$, i.e.,  $\min_{u \in \mathcal U_{los}}\{R^u_{k,b}\}$. Without the  loss of generality, let us assume that
UD $u \in \mathcal U_{los}$ has the minimum rate
that is denoted by $R^u_{min}$. The corresponding transmission duration is $\frac{s}{R^u_{min}}$. 
The design of  $R^u_{min}$  dominates the local models transmission duration from the UDs to the UAVs, thus it dominates the time duration of one FL global iteration. This is because the FL time consists of  the local models transmission time and the learning computation time. Since the computation times of the UDs for local learning does not differ much, the FL time of one global iteration is dominated by $R^u_{min}$. Thus, $R^u_{min}$ can be adapted to include fewer or more UDs in the training process. 

For the different transmission durations $\mathcal U_{los}$, some UDs will finish transmitting their local models before other UDs. Thus, high transmission rate UDs in $\mathcal U_{los}$ will have to wait to start a new iteration simultaneously with relatively good transmission rate UDs. We propose to efficiently exploit such waiting times to assist the UDs that have non-LOS channels to the UAVs.  Define the portion of the time that not being used by $\bar{u}$-th UD (i.e., $\bar{u}\neq u, u \in \mathcal U_{los}$) at the $t$-th iteration is referred
to as the idle time of the $\bar{u}$-th UD and denoted by $T^{\bar{u}}_{idle}$. This idle time can be expressed as $T^{\bar{u}}_{idle}=(\frac{s}{R^{\bar{u}}_{k,b}}-\frac{s}{R^u_{min}})$ seconds.  Such idle time can be exploited by UDs $\bar{u}\in \mathcal U_{los}$  via D2D links if they ensure the complete transmission of the local parameters of the non-LOS UDs to the UAVs. More specifically, the idle time of the $\bar{u}$-th UD should be greater than or equal to the transmission duration of sending the local parameters from the $\hat{u}$-th non-LOS UD  to the $\bar{u}$-th UD plus the time duration of forwarding the local parameters from the $\bar{u}$-th UD to the $k$-th UAV. Mathematically, it must satisfy $T^{\bar{u}}_{idle}\geq (\frac{s}{R^{\bar{u}}_{\hat{u}}}+\frac{s}{R^{\bar{u}}_{k,b}})$. From now on, we will use the term relay to UD   $\bar{u}\neq u, \bar{u} \in \mathcal U_{los}$. In relay mode, each communication period is divided into two intervals corresponding to the non-LOS UD-relay phase (D2D communications) and relay-UAV phase (mmWave communication). The aforementioned transmission duration components of UDs and relays for one global iteration is shown in Fig. \ref{fig2}. Note that UDs can re-use the same frequency band and transmit  simultaneously via
D2D links.

When the $\hat{u}$-th UD does not have a LOS communication to any of the UAVs, it may choose the $\bar{u}$-th UD as its relay if the $\bar{u}$-th relay is located in the coverage zone of the $\hat{u}$-th UD. Let $\mathcal U_{\hat{u}}$ is the set of relays in the coverage zone of UD $\hat{u}$. Let $h^{\bar{u}}_{\hat{u}}$ denote the channel gain for the D2D link between the $\hat{u}$-th UD and the $\bar{u}$-th relay. Then, the achievable rate of D2D pair $(\hat{u}, \bar{u})$ is given by $
R^{\bar{u}}_{\hat{u}}=W\log_{2}(1+\frac{p |h^{\bar{u}}_{\hat{u}}|^2}{N_0}), \forall \bar{u} \in \mathcal U_{los}, \hat{u} \in \mathcal U_{non}$. In relay mode, the transmission duration for sending the local parameter of the $\hat{u}$-th UD to the $k$-th UAV through relay $\bar{u}$  is $\mathtt T_{\hat{u}}=\frac{s}{R^{\bar{u}}_{\hat{u}}}+\frac{s}{R^{\bar{u}}_{k,b}}$, which should satisfy $\mathtt T_{\hat{u}} \leq T^{\bar{u}}_{idle}$.

\begin{figure}[t!]
	\centering
	\includegraphics[width=0.55\linewidth]{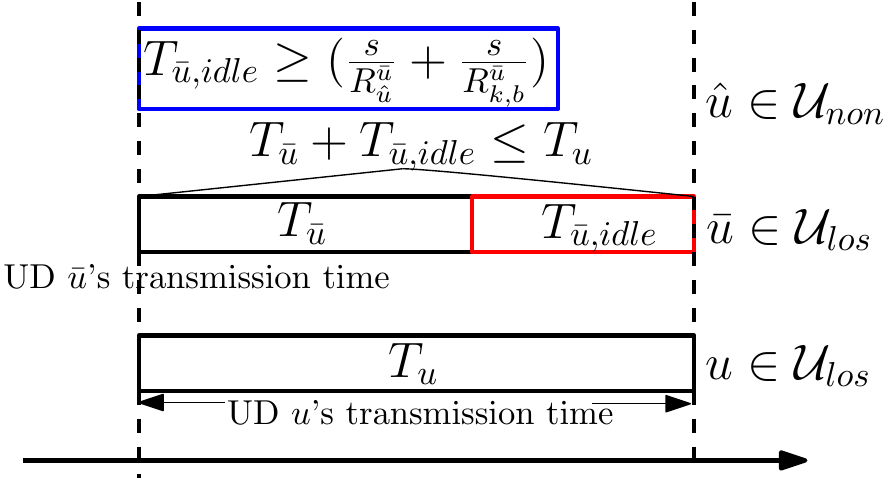}
	\caption{Transmission time structure for LOS UDs and non-LOS UDs
		for the $t$-th global iteration.}
	\label{fig2}
\end{figure}

\vspace{-0.44cm}
\section{FedMoD}
\subsection{Federated Learning Process}
\vspace{-0.25cm}
In FL, each UD $u$ possesses a set of local training data, denoted as $\mathcal D_u$. The local loss function on the dataset of the $u$-th UD can be calculated as
\begin{equation}\label{L_n}
F_u(\mathbf w)=\frac{1}{|\mathcal D_u|}\sum_{(x_i,y_i)\in \mathcal D_u}f_i(\mathbf w), \forall u\in \mathcal U,
\end{equation} 
where $x_i$ is the sample $i$’s input (e.g., image
pixels) and $y_i$ is the sample $i$’s output (e.g., label of the image) and $f_i(\mathbf w)$ is the loss function that measures the local training model error of the $i$-th data sample. 
The collection of data samples at the set of UDs that is associated with the $k$-th UAV 
is denoted as $\mathcal {\tilde{D}}_k$, and the training data at all the learning involved UDs, denoted as $\mathcal U_{inv}$, is denoted as $\mathcal D$. The ratios
of data samples are defined as $\hat{m}_u=\frac{|\mathcal D_u|}{|\mathcal {\tilde{D}}_k|}$, $m_u=\frac{|\mathcal D_u|}{|\mathcal D|}$, and $\tilde{m}_k=\frac{|\mathcal {\tilde{D}}_k|}{|\mathcal D|}$, respectively.
We define the loss function for the $k$-th UAV as the average local loss across the $k$-th cluster $
\hat{F}(\mathbf w)=\sum^{|\mathcal U_k|}_{u=1} \frac{|\mathcal D_u|}{|\mathcal {\tilde{D}}_k|} F_u(\mathbf w)$. The global loss function $F(\mathbf w)$ is then defined as the average loss across all the clusters $
F(\mathbf w)=\sum^{|\mathcal U_{inv}|}_{u=1} \frac{|\mathcal D_u|}{|\mathcal D|} F_u(\mathbf w)$. The objective of the FL model training is to find the optimal model parameters $\mathbf w^*$ for $F(\mathbf w)$ that is expressed as follows $\mathbf w^*=\arg\min_{\mathbf w} F(\mathbf w)$. 
In this work, we propose FedMoD  that involves three main procedures:
1) local model update at the UDs, 2) local model aggregation at the UAVs, and 3) model dissemination between the UAVs.

\textit{1) Local Model Update:} Denote the model of the $u$-th UD at the $t$-th global iteration as $\mathbf w_u(t)$. This UD performs model updating based on its local dataset by using stochastic gradient descent (SGD) algorithm, which is expressed as follows
\begin{equation}
\mathbf w_u(t)=\mathbf w_u(t-1)-\lambda g(\mathbf w_u(t-1)),
\end{equation}
where $\lambda$ is the learning rate and $g(\mathbf w_u(t-1))$ is the stochastic gradient computed on the dataset of the $u$-th UD.

\textit{2) Local Model Aggregation:} After all the selected UDs completing their local model updates,  they offload their model parameters over the available RRBs to the associated UAVs. A typical UAV $k$ aggregates the received models by computing a weighted sum as follows
\begin{equation} \label{w_UAV}
\mathbf {\tilde{w}}_k(t)=\sum_{u\in \mathcal U_k} \hat{m}_u\mathbf w_u(t), \forall k \in \mathcal K.
\end{equation}
\textit{3) Model Dissemination:} Each UAV
disseminates its local aggregated model to the one-hop neighboring UAVs. The model dissemination includes $l=1, 2, \cdots, \alpha$ times of model dissemination until at least one UAV receives the local aggregated models of other UAVs, where $\alpha$ is the number of dissemination rounds. Specifically, at the $t$-th iteration, the $k$-th UAV aggregates the local models of its associated UDs as in \eref{w_UAV}.

At the beginning of the model dissemination step, the  $k$-th UAV knows only $\mathbf {\tilde{w}}_k(t)$ and does not know the models of other UAVs' models $\mathbf {\tilde{w}}_{j}(t), j\neq k, \forall j\in \mathcal K$. Consequently, at the $t$-th global iteration and $l$-th round, the  $k$-th UAV has the following
two sets:
\begin{itemize}
	\item The \textit{Known} local aggregated model: Represented by $\mathcal H^l_k(t)=\{\mathbf {\tilde{w}}_{k}(t)\}$.
	\item The \textit{Unknown} local aggregated models: Represented by $\mathcal W^l_k(t)=\{\mathbf {\tilde{w}}_{j}(t), \mathbf {\tilde{w}}_{\tilde {j}}(t), \cdots, \mathbf {\tilde{w}}_{K}(t)\}$ and defined as the set of the local aggregated models of other UAVs.
\end{itemize}
These two sets are referred as the side information of the UAVs. For instance, at $l=\alpha$, the side information of the $k$-th UAV is $\mathcal H^\alpha_k(t)=\{\mathbf {\tilde{w}}_{k}(t), \mathbf {\tilde{w}}_{j}(t), \mathbf {\tilde{w}}_{\tilde {j}}(t), \cdots, \mathbf {\tilde{w}}_{K}(t)\}$ and $\mathcal W^\alpha_k(t)=\emptyset$. To achieve global model consensus, UAV $k$ needs to know the other UAVs' models, i.e., $\mathcal W_k(t)$, so as to aggregate a global model for the whole network. To this end, we propose  efficient model dissemination scheme that enables the UAVs to obtain their \textit{Unknown} local aggregated models $\mathcal W_k(t), \forall k \in \mathcal K,$ with minimum dissemination latency.

\vspace{-0.25cm}
\subsection{Model Dissemination}\label{Sdiss}
\vspace{-0.25cm}
To overcome the need of CPS for global aggregations or UAV coordination, an efficient distributed model dissemination method is developed. Note that all the associations of
UAVs $\mathcal K_k$ can be computed locally at the $k$-th UAV since all the needed information (e.g., complex channel gains and the indices of the local aggregated models) are locally available. In particular, UAV $k \in \mathcal K$ knows the information of its  neighboring UAVs only. 

At each dissemination round, transmitting UAVs use the previously mentioned side information to perform XOR model encoding, while receiving UAVs need the stored models to obtain the \textit{Unknown} ones. The entire process
of receiving the \textit{Unknown} models takes a small
duration of time. According to the reception
status feedback by each UAV,
the UAVs distributively select the transmitting UAVs and their models to be transmitted to the
receiving UAVs at each round $l$. The transmitted models can be one of the
following two options for each receiving UAV $i$:
\begin{itemize}
	\item Non-innovative model (NIM): A coded model is non-innovative	for the receiving UAV $i$ if it does not contain any model that is
	not known to UAV $i$.
	\item Decodable model (DM): A coded model is
	decodable for the receiving UAV $i$ if it contains just one model that is not known to UAV $i$.
\end{itemize}
In order to represent the XOR coding opportunities among the models not known at
each UAV, we introduce a FedMoD conflict graph. At round $l$, the
FedMoD conflict graph is denoted by $\mathcal G(\mathcal V(l), \mathcal E(l))$, where
$\mathcal V(l)$ refers to the set of vertices, $\mathcal E(l)$ refers to the set of encoding edges. Let $\mathcal K_{k}$ be the set of neighboring UAVs to the $k$-th UAV, and let $\mathcal K_{w}\subset \mathcal K$ be the set of UAVs that still wants some local aggregated models. Hence, the FedMoD graph is designed by generating all vertices for the $k$-th possible UAV transmitter that can provide some models to other UAVs, $\forall k \in \mathcal K$. The vertex set $\mathcal V(l)$ of the entire graph is the union of vertices of all possible transmitting UAVs.
Consider, for now, generating the vertices  of the $k$-th UAV. Note that  the $ k$-th UAV can exploit its previously received models $\mathcal H^l_{k}(t)$ to transmit an encoded/uncoded model to the set of requesting UAVs. Therefore, each
vertex is generated for each model $m\in \mathcal W^l_{i}(t)\cap \mathcal H^l_{k}(t)$ that is requested by each UAV $i\in \mathcal K_{w}\cap \mathcal K_{k}$ and for each achievable rate of the $k$-th UAV $r \in \mathcal R_{k,i} = \{r \in \mathcal R_{k}| r \leq r_{k,i}~ \text{and} ~i\in \mathcal K_w\cap\mathcal K_{k}$\}, where $\mathcal R_{k,i}$ is a set of achievable capacities between the $k$-th UAV and the $i$-th UAV, i.e.,  $\mathcal R_{k,i} \subset \mathcal R_{k}$.
Accordingly, the  $i$-th neighboring UAV in    $ \mathcal K_{k}$ can  receive a model from the $k$-th UAV. Therefore, we generate $|\mathcal R_{k,i}|$ vertices for a requesting  model  $m \in \mathcal H^l_{k}(t)  \cap \mathcal W^l_{i}(t), \forall i \in  \mathcal K_w\cap \mathcal K_{k}$. A vertex $v^{k}_{i,m, r} \in \mathcal V(l)$ indicates the $k$-th UAV can transmit the $m$-th model to the $i$-th UAV with a rate $r$. We define the utility
of vertex $v^{k}_{i,m, r}$ as 
\begin{equation}\label{ww}
w(v^{k}_{i,m, r})=rN_k,
\end{equation}
where $N_k$ is the number of neighboring  UAVs that can be served by the $k$-th UAV. This weight metric shows two potential benefits (i) $N_k$ represents   that the $k$-th transmitting UAV is connected to many other UAVs that are requesting models in $\mathcal H^t_{k}(l)$; and (ii) $r$  provides a balance between  the transmission rate and the number of scheduled  UAVs.


Since UAVs  communicate among them, their connectivity can be characterized by an undirected graph  with  sets of vertices and connections. All possible conflict connections  between vertices (conflict edges between circles) in the FedMod conflict graph are provided as follows. Two vertices $v_{i,m, r}^k$ and $v_{i',m', r'}^{k'}$ are adjacent   by a conflict edge in $\mathcal G$,  if one of the  following conflict conditions (CC) is true.
\begin{itemize}
	\item \textit{\textbf{CC1.} (encoding conflict edge): ($k={k^\prime}$) and ($m\neq m^\prime$) and ($m,m^\prime$) $\notin \mathcal H^l_{k^\prime}(t)\times \mathcal H^l_{k}(t)$. A conflict edge between vertices in the same local FedMoD conflict graph is connected as long as their corresponding are not decodable to a set of scheduled UAVs.}
	\item \textit{\textbf{CC2.} (rate conflict edge): ($k = k^\prime$) and ($k \neq k^\prime$) and ($r\neq r'$). All adjacent vertices  correspond to  the same (or different) UAV $k$ should have the same achievable rate.}
	\item \textit{\textbf{CC3.} (transmission conflict edge): ($k \neq k^\prime$) and ($i=i^\prime$). The same UAV cannot be scheduled to two different UAVs $k$ and $k^\prime$.}
	\item \textit{\textbf{CC4.} (half-duplex conflict edge): ($k = i^\prime$) or ($k^\prime=i$). The same UAV can not transmit and receive in the same  dissemination round.}
\end{itemize}

To distributively disseminate the local aggregated models among the UAVs, we propose a graph theory method as follows. Let $\mathcal{S}_{k}$ represent the associations of the neighboring UAVs in the coverage zone  of the $k$-th UAV, i.e., the associations of UAV $k$ to the set $\mathcal K_k$. Then, let the local  FedMoD conflict graph $\mathcal G_k(\mathcal S_k)\subset \mathcal G$ for an arbitrary UAV $k\in \mathcal K$ represent the set of associations $\mathcal S_k$. Our proposed distributed algorithm has two phases: i) the initial phase and ii) the conflict solution phase. In the initial phase, UAV $k\in \mathcal K$  constructs the local FedMoD conflict graph $\mathcal G_k(\mathcal S_k)$ and selects its targeted neighboring UAVs using the maximum weight independent set (MWIS) search method \cite{TWC, Access} that results in MWIS $\mathbf S_k$. Each UAV exchanges its scheduled UAVs with its neighbor UAV. Then the conflict solution phase starts. The UAV that is associated to multiple UAVs (UAV that is located at the overlapped regions of UAVs) is assigned to one UAV that offers the highest weight of scheduling that UAV. UAVs  that do not offer the maximum weight cannot schedule that UAV, and  therefore remove that UAV from their set of associated UAVs and vertices. We then design the new graph. We repeat this process until all the conflicting UAVs are scheduled to at most a single transmitting  UAV. The details process of the algorithm for a single dissemination round are presented in Algorithm \ref{Algorithm2}. 

\begin{algorithm}[t!]
	\SetAlgoLined
	\caption{Distributed UAV-UAV Scheduling for Model Dissemination}
	\label{Algorithm2}
		\KwData{$\mathcal K$, $\tilde{\mathbf w}_k, \mathcal H^0_k(t), \mathcal W^0_k(t), \forall k\in \mathcal K$.}
		\textbf{Initialize Phase:}\\
	    \textbf{Initialize:}  $\mathtt K=\emptyset$.\\
		\For{\textbf{all} $k\in \mathcal K$}{
		Construct $\mathcal G_k(\mathcal K_k)$ and	calculate weight $w(v)$ using \eref{ww}, $\forall v\in \mathcal G_k$.\\
		Find MWIS $\mathbf S_k$.
	        }
		\textbf{Conflict Solution Phase:}
		\For{$i=1, 2, \cdots$}{
		Transmit $\hat{\mathbf S}_k=\{j\in \mathcal K_k ~|~ j\in \mathbf S_k\}$.\\
		Set $\mathtt K=\{j\in \mathcal K~|~ \exists (k,k')\in \mathcal K^2, j\in \hat{\mathbf S}_k\cap\hat{\mathbf S}_{k'}\}$.\\
		\For{\textbf{all} $j\in \mathtt K$}{
		Set $\hat{\mathcal K}(j)=\{k\in \mathcal K~|~j\in \hat{\mathbf S}_k\}$.\\
		\For{\textbf{all} $k\in \hat{\mathcal K}(j)$}{
		Set $M_{kj}=\sum_{v\in \mathbf S_k}w(v)$ and $\mathcal K_k=\mathcal K_k\backslash \{j\}$.\\
		Construct $\mathcal G_k(\mathcal K_k)$ and compute $w(v)$ by \eref{ww} and
		solve $\tilde{\mathbf S}_k$ MWIS.\\
		Set $\tilde{M}_{kj}=\sum_{v\in \tilde{\mathbf S}_k}w(v)$ and transmit $M_{kj}$ and $\tilde{M}_{kj}$.
	     } 
	    Set $k^*=\arg\max_{k\in \hat{\mathcal K}(j)}\left(M_{kj}+\sum_{k'\in \hat{\mathcal K}(j), k\neq k'}\tilde{M}_{k'j}\right)$.\\
		  Set $\mathcal K_{k^*}=\mathcal K_{k^*}\cup \{j\}$.\\
		\For{\textbf{all} $k\in \hat{\mathcal K}(j)\backslash\{k^*\}$}{
		Set $\mathbf S_k=\tilde{\mathbf S}_k$.
	}
	}
	}
		\KwResult{$\mathbf S=\mathbf S_k, \cdots $}
	 \end{algorithm}

\subsection{Illustration of the Proposed Model Dissemination Method}
For further illustration, we explain the dissemination method that is implemented at the UAVs through an example of the network topology of Fig. \ref{fig3}.  Suppose that all the UAVs have already received the local models of their scheduled UDs and performed the local model averaging. Fig. \ref{fig3} presents the side information status of
each UAV at round $l=0$.

\begin{figure}[t!]
	\centering
	\includegraphics[width=0.95\linewidth]{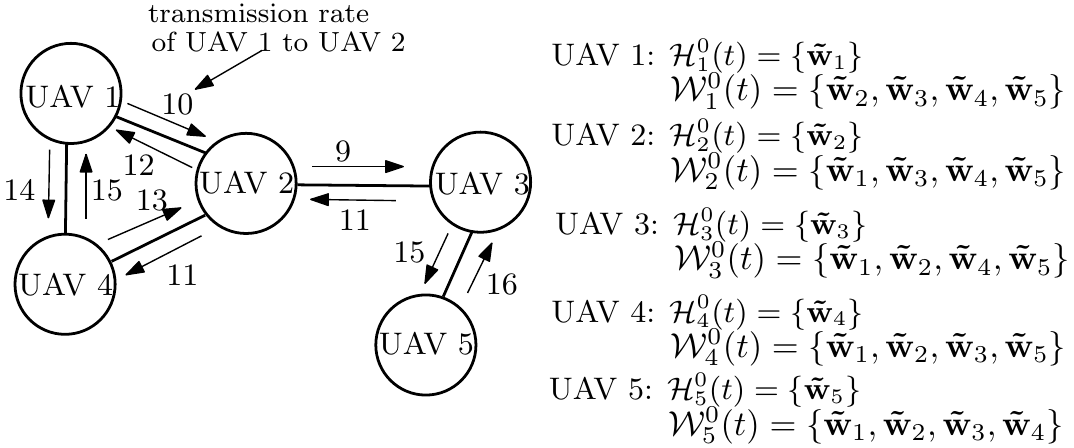}
	\caption{A simple example of $5$ UAVs with their arbitrary transmission rates and initial side information at round $l=0$.}
	\label{fig3}
\end{figure}

\textbf{Round 1:} Since UAV $2$ has good reachability to many UAVs ($\mathcal K_2=\{1, 4, 3\}$), it transmits its model $\mathbf {\tilde{w}}_{2,0}$ to UAVs 1, 4, and 3 with a transmission rate of $r(l=1)=\min\{12, 11, 9\}=9$ Mbps (\textbf{CC2} is satisfied). Note that UAV $5$ can not transmit to UAV $3$ according to \textbf{CC3}, i.e., UAV $3$ is already scheduled to the transmitting UAV $2$. When UAV $2$
finishes model transmission, the \textit{Known} sets of the receiving UAVs is updated to  $\mathcal H^1_1(t)=\{\mathbf {\tilde{w}}_{1}, \mathbf {\tilde{w}}_{2}\}$, $\mathcal H^1_3(t)=\{\mathbf {\tilde{w}}_{3}, \mathbf {\tilde{w}}_{2}\}$, and $\mathcal H^1_4(t)=\{\mathbf {\tilde{w}}_{4}, \mathbf {\tilde{w}}_{2}\}$. Accordingly, their \textit{Unknown} sets are: $\mathcal W^1_1(t)=\{\mathbf {\tilde{w}}_{3}, \mathbf {\tilde{w}}_{4}, \mathbf {\tilde{w}}_{5}\}$,  $\mathcal W^1_3(t)=\{\mathbf {\tilde{w}}_{1}, \mathbf {\tilde{w}}_{4}, \mathbf {\tilde{w}}_{5}\}$,  $\mathcal W^1_4(t)=\{\mathbf {\tilde{w}}_{1}, \mathbf {\tilde{w}}_{3}, \mathbf {\tilde{w}}_{5}\}$.   

\textbf{Round 2:} Although UAV $2$ has good reachability to many UAVs, it would not be selected as a transmitting UAV at $l=2$. This is becasue UAV has already disseminated its side information to the neighboring UAVs, thus UAV $2$ does not have any vertex in the FedMoD conflict graph.  In this case, UAVs $4$ and $5$ can simultaneously transmit models $\mathbf {\tilde{w}}_{4}$ and $\mathbf {\tilde{w}}_{5}$, respectively, to the receiving UAVs $\{1, 2\}$ and $\{3\}$.  When UAVs $4$ and $5$
finish models transmission, the \textit{Known} sets of the receiving UAVs is updated to  $\mathcal H^2_1(t)=\{\mathbf {\tilde{w}}_{1}, \mathbf {\tilde{w}}_{2}, \mathbf {\tilde{w}}_{4}\}$, $\mathcal H^2_2(t)=\{\mathbf {\tilde{w}}_{2}, \mathbf {\tilde{w}}_{4}\}$, and $\mathcal H^2_3(t)=\{\mathbf {\tilde{w}}_{3}, \mathbf {\tilde{w}}_{2}, \mathbf {\tilde{w}}_{5}\}$. Clearly, UAVs $4$ and $5$ transmit their models to the corresponding UAVs with transmission rates of $r_4=\min\{13,15\}=13$ Mbps and $r_5=16$ Mbps, respectively. However, for simultaneous transmission and from \textbf{CC2}, all the vertices of the corresponding UAVs $\{1, 2, 3\}$ should have the same achievable rate. Thus, UAVs $4$ and $5$ adopt one transmission rate which is $r(l=2)=\min\{r_4,r_5\}=13$ Mpbs.

\textbf{Round 3:}  UAV $1$ transmits model $\mathbf {\tilde{w}}_{1}$ to the receiving UAVs $\{2, 4\}$, and their \textit{Known} sets are updated to  $\mathcal H^3_2(t)=\{\mathbf {\tilde{w}}_{2}, \mathbf {\tilde{w}}_{4}, \mathbf {\tilde{w}}_{1}\}$, $\mathcal H^3_4(t)=\{\mathbf {\tilde{w}}_{4}, \mathbf {\tilde{w}}_{2}, \mathbf {\tilde{w}}_{1}\}$. UAV $1$ transmits its model to the corresponding UAVs with a transmission rate of $r(l=3)=\min\{10,14\}=10$ Mbps.

\textbf{Round 4:}  Given the updated side information of the UAVs, UAV $3$ can encode models $\mathbf {\tilde{w}}_{5}$ and $\mathbf {\tilde{w}}_{2}$ into the encoded model $\mathbf {\tilde{w}}_{5} \oplus \mathbf {\tilde{w}}_{2}$ and broadcasts it to UAVs $2$ and $5$.  Upon reception this encoded model, UAV $5$ uses the stored model $\mathbf {\tilde{w}}_{5}$
to complete model decoding $(\mathbf {\tilde{w}}_{5} \oplus \mathbf {\tilde{w}}_{2})\oplus \mathbf {\tilde{w}}_{5} =\mathbf {\tilde{w}}_{2}$. Similarly, UAV $5$ uses the stored model $\mathbf {\tilde{w}}_{2}$
to complete model decoding $(\mathbf {\tilde{w}}_{5} \oplus \mathbf {\tilde{w}}_{2})\oplus \mathbf {\tilde{w}}_{2} =\mathbf {\tilde{w}}_{5}$. The broadcasted model is thus decodable for both UAVs $5$ and $2$ and has been transmitted with a rate of $r(l=4)=\min\{11,15\}=11$ Mpbs. The \textit{Known} sets of these receiving UAVs are as follows: $\mathcal H^4_2(t)=\{\mathbf {\tilde{w}}_{2}, \mathbf {\tilde{w}}_{4}, \mathbf {\tilde{w}}_{1}, \mathbf {\tilde{w}}_{5}\}$ and  $\mathcal H^4_5(t)=\{\mathbf {\tilde{w}}_{5}, \mathbf {\tilde{w}}_{2}\}$.

\textbf{Round 5:}  Given the updated side information of the UAVs at $l=4$, UAV $3$ transmits $\mathbf {\tilde{w}}_{3}$ to UAVs $2$ and $5$.  Upon reception this model, UAV $2$ has obtained all the required models, i.e., $\mathcal H^5_2(t)=\{\mathbf {\tilde{w}}_{1}, \mathbf {\tilde{w}}_{2}, \mathbf {\tilde{w}}_{3}, \mathbf {\tilde{w}}_{4}, \mathbf {\tilde{w}}_{5}\}$ and $\mathcal W^5_2(t)=\{\emptyset\}$.  The broadcasted model is transmitted with a rate of $r(l=5)=\min\{11,15\}=11$ Mpbs. Since UAV $2$ has all the local aggregated models of other UAVs, it can aggregate them all which results the global model at the $t$-th iteration: 
\begin{equation}\label{7eq}
\mathbf {\tilde{w}}(t)=\frac{1}{D}(\mathbf {\tilde{w}}_{1}+\mathbf {\tilde{w}}_{2}+\mathbf {\tilde{w}}_{3}+\mathbf {\tilde{w}}_{4}+\mathbf {\tilde{w}}_{5}).
\end{equation}
Therefore, the global model $\mathbf {\tilde{w}}$ is broadcasted from UAV $2$ to UAVs $\{1, 4, 3\}$ with a rate of $\min\{12, 11, 9\}=9$ Mpbs. Next, UAV $3$ can send $\mathbf {\tilde{w}}$ to UAV $5$ with a rate of $15$ Mbps. Therefore, all the UAVs obtain the shared global model $\mathbf {\tilde{w}}$ and broadcast it to their scheduled UDs to initialize the next iteration $t+1$. Note that the transmission duration of these dissemination rounds is 
\begin{equation}\label{diss}
T_{diss}=\underbrace{\frac{s}{9}}_{l=1}+\underbrace{\frac{s}{13}}_{l=2}+\underbrace{\frac{s}{10}}_{l=3}+\underbrace{\frac{s}{11}}_{l=4}+\underbrace{\frac{s}{11}}_{l=5}+\underbrace{\frac{s}{9}+\frac{s}{15}.}_{\mathbf {\tilde{w}}~ \text{broadcasting}}
\end{equation}
The size of a typical model is $s=9.098$ Kb \cite{saif_new, FL_Lit_1_4, FL3}, thus $T_{diss}=0.0059$ sec. Thanks to the efficient model dissemination proposed method that disseminates models from transmitting UAVs to the closest receiving UAVs with good connectivity, the dissemination delay is negligible.\\
\textit{\textbf{Remark 1:}} \textit{In the fully connected model, each UAV can receive the local aggregated models of all UAVs in $K$ dissemination rounds, where each UAV takes a round for broadcasting its local aggregated model to other UAVs.}

The steps of FedMoD that includes local model update, local aggregation at the UAVs, and model dissemination among the UAVs are summarized in Algorithm \ref{alg1}. 
\begin{algorithm}[t!]
	\SetAlgoLined
	\KwData{Number of  global iterations $T$, number of local
		iterations $T_l$}
	\textbf{Initialize:} $t=1$ and start with the same model for each UD $u$: $\mathbf w_{u}(t-1)$.\\
			\For{$t=1, 2, \cdots, T$}{
				\For{each UD $u \in \mathcal U_{inv}$ in parallel}{
				Update the local model as $\mathbf w_{u}(t)$ according to (5).
							}
			\For{each UAV $k \in \mathcal K$ in parallel}{
				Receive the most updated model from the UDs in $\mathcal U_k$.\\
				Obtain $\mathbf {\tilde{w}}_k(t)$  by performing local model aggregation according to \eref{w_UAV}.\\
				\For{$l=1, 2, \cdots, \alpha$}{
					UAVs dissiminate their models among them as expalined in Section III-B and Algorithm \ref{Algorithm2}.\\
				}}
			Update $\mathbf {\tilde{w}}_k(t-1)=\mathbf {\tilde{w}}_k(t)=\mathbf {w}(t)$.\\
			Broadcast $\mathbf {w}(t)$
		to the UDs in $\mathcal U_k$.
			
		Update $t=t+1$.
		}

	\KwResult{Final global model $\mathbf w$.}
	
	\caption{FedMoD Algorithm} \label{alg1}
	
\end{algorithm}
\vspace{-0.55cm}
\subsection{Convergence Analysis}\label{CA}
\vspace{-0.25cm}
In this sub-section, we prove the convergence of FedMoD. To facilitate the convergence rate analysis of the proposed scheme, we first provide the following assumptions. For all $u\in \mathcal U$, we assume: 
\begin{enumerate}
	\item The local loss function is $L$-smooth, i.e.,. This assumption implies that for some $L >0$, 
	$\lVert  \triangledown  F_u(\mathbf w(t+1))-\triangledown  F_u(\mathbf w(t))  \rVert_2 \leq L \lVert  \mathbf w(t+1)-\mathbf w(t)  \rVert_2$.
	\item The mini-batch gradient is unbiased, i.e., $
	\mathbb E_{\mathcal D_u|\tilde{\mathbf w}}[f(\mathcal D_u;\tilde{\mathbf w})]=\triangledown F_u(\tilde{\mathbf w}),$	and there exists $\sigma >0$ such that $
	\mathbb E_{\mathcal D_u|\tilde{\mathbf w}}\bigg \lVert [f(\mathcal D_u;\tilde{\mathbf w})]-\triangledown F_u(\tilde{\mathbf w})\bigg \lVert^2_2 \leq \sigma^2$.
	\item For the degree of non-IIDness, we assume that there exists $\kappa>0$ such that $
	\lVert  \triangledown F_u(\tilde{\mathbf w})- \triangledown F(\tilde{\mathbf w}) \lVert_2 \leq \kappa,$	where $\kappa$ measures the degree of data heterogeneity across all UDs.
\end{enumerate}

In centralized FL, the global model at the CPS at each global iteration evolves according to the following expression \cite{FL_Lit_1_4}:
\begin{align}\label{ag3}
\mathbf{w}(t+1)= \mathbf{w}(t)-\lambda \mathbf G(t),
\end{align}
where $\mathbf{w}(t)=[\mathbf w_u(t)]_{u\in \mathcal U_{inv}}$ and $\mathbf G(t)= [g(\mathbf w_u(t))]_{u\in \mathcal U_{inv}}$. However, in FedMoD, the $k$-th UAV maintains a model updated based on the trained models of its scheduled UDs only and needs to aggregate the models of other UAVs using the model dissemination method as in \sref{Sdiss}. Theretofore, each UAV has insufficient model averaging unless the model dissemination method is performed until all UAVs obtain the global model defined in \eref{7eq}, i.e., at $l=\alpha$. In other words, at $l=\alpha$, the global model of our proposed decentralized FL should be the one mentioned in \eref{ag3}. For convenience, we define $\mathbf {\tilde{u}}(t)=\sum_{u\in \mathcal U_{inv}}m_u \mathbf w_u(t)$, and consequently, $\mathbf {\tilde{u}}(t)=\mathbf{\tilde{w}}(t)\mathbf m$. By multiplying both sides of the evolution expression in \eref{ag3} by $\mathbf m$, yielding the following expression
\begin{align}\label{agg3}
\mathbf{\tilde{u}}(t+1)= \mathbf{\tilde{u}}(t)-\lambda \mathbf G(t)\mathbf m,
\end{align}
Following \cite{FL_Dec_1_2, FL_Dec_1_3} and leveraging the evolution expression of
$\mathbf{\tilde{u}}(t)$ in \eref{agg3}, we bound the expected change of the local loss
functions in consecutive iterations as follows.
\begin{lemma}\label{lemma1}
	The expected change of the global loss function in two consecutive iterations can be bounded as follows
	\begin{align}\label{eq17} 
	&\mathbb{E}[F(\mathbf{\tilde{u}}(t+1))]-\mathbb{E}[F(\mathbf{\tilde{u}}(t))]  \nonumber \leq \frac{-\lambda}{2} \mathbb E  \lVert \triangledown F(\mathbf{\tilde{u}}(t)) \lVert^2_2\\ &+ \frac{\lambda^2 L}{2} \sum_{u=1}^{{U_{inv}}} m_u\sigma^2 -\frac{\lambda}{2}(1-\lambda L)\tilde{Q}+\frac{\lambda L^2}{2}\mathbb E \bigg \lVert \mathbf{\tilde{w}}(t)(\mathbf I-\mathbf M) \bigg \lVert^2_\mathbf M,
	\end{align}
	where $\tilde{Q}=	\mathbb{E}\bigg [ \bigg \lVert \sum_{u=1}^{U_{inv}} m_u \triangledown F_u(\mathbf w_u(t)) \bigg \lVert^2_2 \bigg ]$, $\mathbf M=\mathbf m \mathbf I^T$, and $\lVert \mathbf X \lVert_{\mathbf M}=\sum_{i=1}^M\sum_{j=1}^N m_{i,j}|x_{i,j}|^2$ is the
	weighted Frobenius norm of an $M\times N$ matrix $\mathbf X$.
\end{lemma}
For proof, please refer to Appendix A.

Notice that $\mathbf{\tilde{w}}(t)$ deviates from the desired global model due to the partial connectivity of the UAVs that results in the last term in the right-hand side (RHS) of
\eref{eq17}. However, through the model dissemination method and at $l=\alpha$, FedMoD ensures that each UAV can aggregate the models of the whole network at each global iteration before proceeding to the next iteration. Thus, such deviation is eliminated. 

Due to the model dissemination among the UAVs, there is a dissemination gap that is denoted by the \textit{dissemination gap} between the $k$-th and $j$-th UAVs as $\delta_{j,k}(t)$, which is the number of dissemination steps that the local aggregated model of the $j$-th UAV needs to be transmitted to  the $k$-th UAV. For illustration, consider the example in Fig. 3, the highest dissemination gap is the one  between UAVs $5$ and $1$ which is $3$. Thus, $\delta_{5,1}(t)=3$. The maximum dissemination gap of UAV $k$ is $\delta_k(t)=\max_{j\in \mathcal K}\{\delta_{j,k}(t)\}$. Therefore, a larger value of $\delta_{j,k}(t)$ implies that the model of each UAV needs more dissemination step to be globally converged. The following remark shows that $\delta_k(t)$ is upper bounded throughout the whole training process.\\
\textit{\textbf{Remark 2:}} \textit{There exists a constant $\delta_{max}$ such that $\delta_k(t)\leq \delta_{max}$, $\forall t\in T, k\in \mathcal K$. At any iteration $t$, the dissemination gap of the farthest UAV (i.e., the UAV at the network edge), $\delta_{max}= \alpha$ gives a maximal		value for the steps that the models of other UAVs have been disseminated to UAV $k$.}

Given the aforementioned analysis, we are now ready to prove the convergence of FedMoD.
	\begin{theorem}\label{th3} 
		If the learning rate $\lambda$ satisfies $
		1-\lambda L \geq 0, 1-2\lambda^2 L^2 >0$,	we have 
		\begin{align}\small \label{final_eq}
		\mathbb E [\lVert \triangledown F(\tilde{\mathbf u})(t)\lVert^2_2]\leq \frac{2\{\mathbb E[F(\tilde{\mathbf u})(0)-F(\tilde{\mathbf u})(T)]\}}{\delta }+\lambda L\sum_{u=1}^{U_{inv}}m_u\sigma^2
		\end{align}
	\end{theorem} 
	\begin{proof} 
		From \eqref{eq17}, we have
		\begin{align} \label{27eq}\nonumber
			\frac{\lambda}{2} \mathbb E \lVert \triangledown F(\mathbf{\tilde{u}}(t)) \lVert^2_2 &\leq \mathbb{E}[F(\mathbf{\tilde{u}}(t))]- \mathbb{E}[F(\mathbf{\tilde{u}}(t+1))]\\&+\frac{\lambda^2 L}{2} \sum_{u=1}^{U_{inv}} m_u\sigma^2  -\frac{\lambda}{2}(1-\lambda L)\tilde{Q}.
		\end{align}	
		\begin{align} \label{2777eq}
		\mathbb E \lVert \triangledown F(\mathbf{\tilde{u}}(t)) \lVert^2 \nonumber & \leq \frac{2\{\mathbb{E}[F(\mathbf{\tilde{u}}(t))]- \mathbb{E}[F(\mathbf{\tilde{u}}(t+1))]\}}{\lambda} \\&+\lambda L \sum_{i=1}^{U_{inv}} m_u\sigma^2  -(1-\lambda L)\tilde{Q}
		\end{align}
		Since $1-\lambda L \geq 0$ from \thref{th3}, the third term in the RHS of \eqref{2777eq} is eliminated, thus we have
		\begin{align} \label{27eq_u}
		\mathbb E [\lVert \triangledown F(\tilde{\mathbf u})(t)\lVert^2]\leq \frac{2\{\mathbb E[F(\tilde{\mathbf u})(0)-F(\tilde{\mathbf u})(T)]\}}{\lambda}+\lambda L\sum_{u=1}^{U_{inv}}m_u\sigma^2
		\end{align}
\end{proof}

\vspace{-0.99cm} 
\section{FedMoD: Modeling and Problem Formulation}
\subsection{FL Time and  Energy Consumption}
\vspace{-0.22cm}
\textit{1) FL Time:} 
The constrained FL time at each global iteration  consists of both computation and wireless transmission time that is explained below.

The wireless transmission time consists of (1) the 	uplink transmission time for transmitting the local updates from the UDs to the associated UAVs $\mathcal K$. This transmission time is already discussed in Section II-C and represented by $T_u$. (2) The  transmission time for disseminating the local aggregated models among the UAVs. The model dissemination time among all the UAVs is $T_{diss}$ as given in \eref{diss}. (3) The downlink transmission time for transmitting the local aggregated models from the UAVs to the scheduled UDs $\mathcal U$. The downlink transmission time for UAV $k$ can be expressed $T^{do}_k=\frac{s}{R_k}$. On the other hand, the computation time for local learning at the $u$-th UD is expressed as $
T^{comp}_u=T_l\frac{Q_uD_u}{f_u},$ where $T_l$ is the number of local iterations to reach the local accuracy $\epsilon_l$ in the $u$-th UD, $Q_u$ as the number of CPU cycles to process one data sample, and $f_u$ is the computational frequency of the CPU in the $u$-th UD (in cycles per second).

By combining the aforementioned components, the FL time $\tau_k$ at the $k$-th UAV can be calculated as
	\begin{align}
		\tau_k & \nonumber=  \max_{u \in \mathcal U_{k}}T^{comp}_u+\max_{u \in \mathcal U_k} T^{com}_u+T^{do}_k \\&=\max_{u \in \mathcal U_k} \left\{T_l\frac{Q_u D_n}{f_u}\right\}+\max_{u \in \mathcal N_k}\left\{\frac{s}{R_{k,b}^u}\right\}+\frac{s}{R_k}.
	\end{align}  
Therefore, the total FL time over all global iterations $T$ is $\tau=T(\max_{k\in \mathcal K}(\tau_k)+T_{diss})$, which should be no more than the maximum FL time threshold $T_\text{max}$. This constraint, over all global iterations $T$, is expressed as
	\begin{align}\label{17} 
	\tau \nonumber=T\bigg (\underbrace{\max_{u \in \mathcal N} \left\{T_l\frac{Q_n D_u}{f_u}\right\}}_\text{local learning}+\underbrace{T_u}_\text{uplink transmission}+\underbrace{\max_{k \in \mathcal K}\left\{\frac{s}{R_k}\right\}}_\text{downlink transmission}\\  +\underbrace{T_{diss}}_\text{dissemination duration}\bigg)\leq T_\text{max}.
	\end{align}

\textit{2) Energy Consumption:} The system’s energy is consumed for local model training at the UDs, wireless models transmission, and UAVs' hovering in the air.

\textit{1) Local computation:} The well-known energy consumption model for the local computation is considered, where the energy 	consumption of the $u$-th UD to process a single CPU cycle is $\alpha f^2_u$, and $\alpha$	is a constant related to the switched capacitance \cite{CPU2, CPU3}. Thus, the energy consumption of the $u$-th UD	for local computation is $
	E^{comp}_u=T_{loc}C_uD_u\alpha f^2_u.$

\textit{2) Wireless models transmission:} The energy consumption to transmit the local model parameters to the associated UAVs can be denoted by $E^{com}_u$ and calculated as $	E^{com}_u=P_uT^{com}_u$. Then, the total energy consumption  $E_u$ at the $u$-th UD is $E_u=E^{comp}_u+E^{com}_u$. 
In a similar manner, the consumed energy for transmitting the local aggregated models	back to the associated UDs can be denoted by $E^{com}_k$ and calculated as $	E^{com}_k=PT^{com}_k$.

\textit{3) UAV's hovering energy:}  UAVs need  to remain stationary in the air, thus most of UAV's energy is consumed for hovering. The
UAV’s hovering power is expressed as \cite{zoheb} $
	p^{hov}=\sqrt{\frac{(mg)^3}{2 \pi r_p^2 n_p \rho}},$	where $m$ is UAV's weight, $g$ is the   gravitational	acceleration of the earth, $r_p$ is propellers' radius, $n_p$ is the number	of propellers, and $\rho$ is the air density. In general, these	parameters of all the UAVs are the same. The hovering time of the $k$-th UAV in each global iteration depends on $\tau$. Hence, the hovering energy of the $k$-th UAV 	can be calculated as $	E^{hov}_k=p^{hov}\tau^t$. In summary, the overall energy consumption of the $k$-th UAV and the $u$-th UD, respectively, are
\begin{align}\label{EE}
E_k=T\left\{E^{hov}_l+E^{com}_l\right\},  E_u=T\left\{E^{comp}_u+E^{com}_u\right\}.
\end{align}


\vspace{-0.4cm}
\subsection{Problem Formulation}
\vspace{-0.22cm}
Given the ATIN and its FL time and energy components, our next step is to formulate the energy consumption minimization problem that involves the joint optimization of two sub-problems, namely UAV-LOS UD clustering and D2D scheduling sub-problems.  To minimize the energy consumption at each global iteration, we need to develop a framework that decides: i) the UAV-UD clustering; ii) the adopted transmission rate of the UDs $\mathcal U_{los}$  to transmit their local models to the set of UAVs/RRBs; and iii)  the set of D2D transmitters (relays) that helping the non-LOS UDs to transmit their local models to the set of UAVs $\mathcal K$. As such, the local models are delivered to all UAVs with minimum duration time, thus minimum energy consumption for UAV's hovering and UD's wireless transmission. Therefore, the energy consumption minimization problem in
the ATIN can be formulated as follows.
\begin{subequations} 
	\begin{align} \nonumber 
	& \text{P0:} 
	\min_{\substack{R^u_{min}, \mathcal U_{los}, \mathcal U_{non}}}  \sum_{k\in \mathcal K}E_k + \sum_{u\in \mathcal U}E_u\\
	&\rm s.t.
	\begin{cases}  \nonumber
	\text{C1:}\hspace{0.2cm} \mathcal U_{k,los}\cap \mathcal U_{k',los}=\emptyset, \forall (k, k')\in \mathcal K, \\
	\text{C2:}\hspace{0.2cm} \mathcal U_{k,los}\cap \mathcal U_{u',los}=\emptyset, \forall k\in \mathcal K,  \\ 
	\text{C3:}\hspace{0.2cm} \mathcal U_{u,los}\cap \mathcal U_{u'',los}=\emptyset, \\
	\text{C4:}\hspace{0.2cm} R^u_{k,b} \geq R_0, (u,k,b)\in (\mathcal U, \mathcal K, \mathcal B),\\
	 \text{C5:}\hspace{0.2cm}  \mathtt T_{\hat{u}} \leq T_u, u \in \mathcal U_{los},\\
	 \text{C6:}\hspace{0.2cm}  T^{\bar{u}}_{idle}\geq (\frac{s}{R^{\bar{u}}_{\hat{u}}}+\frac{s}{R^{\bar{u}}_{k,b}}), \bar{u} \in \mathcal U_{los}, \hat{u} \in \mathcal U_{non}, \\
   \text{C7:}\hspace{0.2cm}  \tau\leq T_\text{max}. 
		\end{cases}
	\end{align}
\end{subequations}
The constraints are explained as follows. Constraint C1 states that the set of scheduled UDs to the UAVs are disjoint, i.e., each UD
must be scheduled to only one UAV. Constraints C2 and C3 make sure that each
UD can be scheduled to only one relay
and no user can be scheduled to a relay and UAV
at the same time instant. Constraint C4 is on the coverage threshold of each UAV. Constraint C5 ensures that the local parameters of UD $\hat{u}$  has to be delivered to UAV $k$ via relay $\bar{u}$ within $\frac{s}{R^u_{min}}$, i.e., $\mathtt T_{\hat{u}}=\frac{s}{R^{\bar{u}}_{\hat{u}}}+\frac{s}{R^{\bar{u}}_{k,b}} \leq \frac{s}{R^u_{min}}$. Constraint C6 ensures that the idle time of UD $\bar{u}$ is long enough for transmitting the local parameters of UD $\hat{u}$ to UAV $k$. Constraint C7 is for the FL time threshold $T_{max}$. We can readily show that problem P0 is NP-hard. However, by analyzing the problem, we can decompose it into two sub-problems and solve them individually and efficiently.

\vspace{-0.4cm}
\subsection{Problem Decomposition}
\vspace{-0.25cm}
First, we focus on minimizing the energy consumption via efficient RRM scheduling of UDs $\mathcal U_{los}$ to the UAVs/RRBs. In particular, we can get the possible minimum transmission duration of UD $u \in \mathcal U_{los}$ by jointly optimizing the UD scheduling and rate adaptation  in $\mathcal U_{los}$. The mathematical formulation	for minimizing the energy consumption via minimizing the transmission durations for UDs-UAVs/RRBs 	transmissions can be expressed as
\begin{subequations} 
	\begin{align} \nonumber 
	& \text{P1:} 
	\min_{\substack{R^u_{min}, \mathcal U_{los}}}  \sum_{k\in \mathcal K}E_k + \sum_{u\in \mathcal U}E_u\\
	&\rm s.t.
	\begin{cases}  \nonumber
	(\text{C1}), \hspace{0.2cm} (\text{C4}), \hspace{0.2cm} (\text{C5}), \hspace{0.2cm} (\text{C7}).
	\end{cases}
	\end{align}
\end{subequations}
Note that this sub-problem contains UD-UAV/RRB scheduling and an efficient solution will be developed in Section IV-B.

After obtaining the possible transmission duration from
UD-UAV transmissions, denoted by $T_u$ of the $u$-th UD ($u\in \mathcal U_{los}$), by solving P1, we can now formulate the second sub-problem.  In particular, we can minimize the energy consumption of non-LOS UDs
$\mathcal U_{non}$ that are not been scheduled to the UAVs within $T_u$ by using D2D communications via relaying mode.  For this, UDs being scheduled to the UAVs from sub-problem P1 can be exploited to work as relays and schedule non-LOS UDs on D2D links within their idle
times. Therefore, the second sub-problem of
minimizing the energy consumption of unscheduled UDs to be scheduled on D2D
links via relaying mode can be expressed as P2 as follows 
\begin{subequations} 
	\begin{align} \nonumber 
	& \text{P2:} 
	\min_{\substack{\mathcal U_{non}}} \sum_{k\in \mathcal K}E_k + \sum_{u\in \mathcal U}E_u\\
	&\rm s.t.
	\begin{cases}  \nonumber
	(\text{C2}), \hspace{0.2cm} (\text{C3}),\hspace{0.2cm}
	(\text{C5}), \hspace{0.2cm} (\text{C6}),
	\text{C8:} \hspace{0.2cm} \mathcal N_{non}\in \mathcal P(\mathcal U \backslash \mathcal N_{los}). \\
	\end{cases}
	\end{align}
\end{subequations}
Constraint C8 states that the set of relays
is constrained only on the UDs that are not been scheduled to the UAVs. 
It can be easily observed that P2 is a D2D scheduling problem that considers selection of relays and their non-LOS scheduled UDs.

\vspace{-0.66cm}
\section{FedMoD: Proposed Solution}

\subsection{Solution to Subproblem P1: UAV-UD Clustering}\label{CCC}
Let $\mathcal{A}$ denote the set of all possible associations between UAVs, RRBs, and LOS UDs, i.e., $\mathcal{A}=\mathcal{K}\times\mathcal{Z}\times \mathcal{U}_{los}$. For instance, one possible association $a$ in $\mathcal A$ is $(k,z,u)$ which represents UAV $k$, RRB $z$, and  UD $u$. Let the conflict clustering graph  in the network is denoted by $\mathcal{G}(\mathcal{V},\mathcal{E})$ wherein $\mathcal{V}$ and $\mathcal{E}$ are the sets of vertices and edges of $\mathcal{G}$, respectively. A typical vertex in $\mathcal{G}$ represents an association in $\mathcal A$, and each edge between two different vertices represents a conflict connection between
the two corresponding associations of the vertices according to C1 in $\text{P1}$. Therefore, we construct the conflict
	clustering graph by generating a vertex $v \in \mathcal{V}$ associated with $a\in \mathcal{A}$ for UDs who have enough energy for performing learning and wireless transmissions. To select the UD-UAV/RRB scheduling that provides a minimum energy consumption while ensures C4 and C7 in $\text{P1}$, a weight $w(v)$ is assigned to each vertex $v \in \mathcal V$. For simplicity, we define the weight of vertex $v^{z}_{k,u}$ as $w(v^{z}_{k,u})=E^{comp}_u+E^{com}_u$.  Vertices $v^{z}_{k,u}$ and $v^{z'}_{k',u'}$ are conflicting vertices that will be connected by an edge in $\mathcal{E}$ if one of the below \textbf{connectivity conditions (CC)} is satisfied:
\begin{itemize}
	\item \textbf{CC1:} ($u=u'$ and $z=z'$ or $k=k'$). \textbf{CC1} states that  the same  user $u$ is in both vertices $v^{z}_{k,u}$ and $v^{z'}_{k',u'}$.
	\item \textbf{CC2:} ($z=z'$ and $u \neq u'$). \textbf{CC2} implies that the same RRB is in both vertices $v^{z}_{k,u}$ and $v^{z'}_{k,u'}$. 
\end{itemize}
Clearly, \textbf{CC1} and \textbf{CC2} correspond to a violation of the constraint C1 of  $\text{P1}$ where two vertices are conflicting if: (i) the UD
is associated with two different UAVs and(or) two different RRBs; or (ii) the RRB is associated with more than one UD. With the designed conflict clustering graph, $\text P1$ is similar	to MWIS problems in several aspects. In MWIS, two		vertices should be non-adjacent in the graph (conditions \textbf{CC1} and \textbf{CC2} in \sref{CCC}), and similarly, in $\text{P1}$, same local learning user cannot be scheduled to	two different UAVs or two different RRBs (i.e., C1). Moreover, the objective of	problem $\text{P1}$ is to minimize the energy consumption, and similarly, the goal of MWIS is to select a		number of vertices that have small weights. Therefore, the following theorem characterizes the
solution to the energy consumption minimization problem P2 in an ATIN.\\
\textbf{Theorem 2.} 
\textit{The  solution to problem P1 is equivalent to the minimum independent set weighting-search method, in which the weight of each
	vertex $v$ corresponding to UD $u$ is
	\begin{align} \label{eqweight} 
	w(v)=E^{comp}_u+E^{com}_u.
	\end{align}	}
Finding the minimum weight independent set $\Gamma^*$  among all other minimal sets  in $\mathcal G$ graph is explained as follows. First, we select vertex $v_i\in \mathcal V, (i=1,2,\cdots, )$  that has the minimum weight $w (v^*_i)$  and add it to $\Gamma^*$ (at this point $\Gamma^* = \{v^*_i\}$). Then,  the  subgraph $\mathcal G(\Gamma^*)$, which consists of vertices in graph  $\mathcal G$ that are not adjacent to vertex $v^*_i$, is extracted and considered for the next vertex selection process. Second, we select a new minimum weight  vertex $v^*_{i'}$ (i.e., $v^*_{i'}$ should be in the corresponding set of $v^*_i$) from subgraph $\mathcal G(\Gamma^*)$. Now, $\Gamma^* = \{v^*_i, v^*_{i'}\}$. We repeat this process until no further vertex is not adjacent  to all vertices in $\Gamma^*$. The selected $\mathtt C$ contains at most $ZK$ vertices. Essentially, any possible solution $\Gamma^*=\{v^*_1, v^*_2, \cdots, v^*_{ZK}\}$ to P1 represents a feasible UD-UAV/RRB scheduling.
\vspace{-0.33cm}
\subsection{Solution to Subproblem P2: D2D Graph Construction}\label{CCC}
In this subsection, our main focus is to schedule the non-LOS UDs to the LOS UDs (relays)  over their idle times so as the local models of those non-LOS UDs can be forwarded to the UAVs. Since non-LOS UDs  communicate with their respective relays over  D2D links, the D2D connectivity can be characterized by an undirected graph $\mathcal G(\mathcal V, \mathcal E)$ with $\mathcal V$ denoting the set of vertices and $\mathcal E$ the set of edges. We construct a new D2D conflict graph  that considers all possible conflicts for scheduling non-LOS UDs on D2D links, such as transmission and half-duplex conflicts. This leads to feasible transmissions from  the potential D2D transmitters $|\mathcal U_{\text{non,tra}}|$.

Recall $\mathcal U_{\text{non}}$ is the set of non-LOS UDs, i.e., $\mathcal U_{\text{non}}=\mathcal U\backslash \mathcal U_{los}$, and let $\mathcal U_{\text{relay}}= \mathcal U_{\text{los}}\backslash\{u\}$ denote the set of relays  that can use their idle times to help the non-LOS UDs. Hence, the D2D conflict graph is designed by generating all vertices for $\bar{u}$-th possible relay, $\forall \bar{u} \in \mathcal U_{\text{relay}}$. The vertex set $\mathcal V$ of the entire graph is the union of vertices of all users.
Consider, for now, generating the vertices  of the $\bar{u}$-th relay. Note that  $\bar{u}$-th relay can help one non-LOS UD as long as it is in the coverage zone is capable of delivering the local model to the scheduled UAV within its idle time. Therefore, each
vertex is generated for each single non-LoS UD that is located in the coverage zone of the $\bar{u}$-th relay and $\mathtt T_{\hat{u}} \leq T^{\bar{u}}_{idle}$. Accordingly, $i$-th non-LOS UD in the coverage zone $ \mathcal Z_{\bar{u}}$ can  transmit its model to the $\bar{u}$-th relay. Therefore, we generate $|\mathcal Z_{\bar{u}}|$ vertices for the $\bar{u}$-th relay.

All possible conflict connections  between vertices (conflict edges between circles) in the D2D conflict graph are provided as follows. Two vertices $v_{i}^{\bar{u}}$ and $v_{i'}^{u'}$ are adjacent   by a conflict edge in $\mathcal G_\text{d2d}$,  if one of the  following conflict conditions is true: (i)  ($\bar{u} \neq u^\prime$) and ($i=i^\prime$). The same non-LoS UD cannot be scheduled to two different helpers $\bar{u}$ and $u^\prime$. (ii) ($i \neq i^\prime$) and ($\bar{u}=u^\prime$). Tow different non-LoS UDs can not be scheduled to the same relay. These two conditions represent C3 in P2, where each non-LoS UD must be assigned to one relay and the same relay cannot accommodate more than one non-LoS UD.  Given the aforementioned designed D2D conflict graph, the following
theorem reformulates the subproblem P3.\\
\textbf{Theorem 3.} 
\textit{The subproblem of scheduling non-LOS UDs on D2D links in $P2$ is equivalently represented by the MWIS selection  among all the maximal sets
in the $\mathcal G_\text{d2d}$ graph, where the weight $\psi(v_{i}^{\bar{u}})$ of each vertex $v_{i}^{\bar{u}}$ is given by
$\psi(v_{i}^{\bar{u}}) =r.$}


\vspace{-0.4cm}
\section{Numerical Results}
\vspace{-0.25cm}
For our simulations,  a circular network area having a radius of $400$ meter (m) is considered. The height of the CPS is $10$ m \cite{zoheb}. Unless specified otherwise, we divide the considered circular network area  into $5$ target locations. As mentioned in the system model, each target location is assigned to one UAV where the locations of the UAVs are randomly distributed in the flying plane with  altitude of $100$ m. The users are placed randomly in the area.  In addition, $U$ users are connected to the UAVs through orthogonal RRBs for uplink local model transmissions. The bandwidth of each RRB is $2$ MHz. The UAV communicates with the neighboring UAVs via high-speed mmWave communication links \cite{17, FL_Dec_1_2}. 

Our proposed FedMod scheme is evaluated on the MNIST and CIFAR-10 datasets, which are well-known benchmark datasets for image classification tasks. Each image is one of $10$ categories. We divide the
dataset into the UDs’ local data $\mathcal D_u$ with  non-i.i.d. data heterogeneity, where each local dataset contains datapoints from two of the $10$ labels. In each case, $\mathcal D_u$ is selected randomly from the full dataset of labels assigned to $u$-th UD. We also assume non-iid-clustering, where the maximum number of assigned classes for each cluster is $6$ classes. For ML models, we use a deep neural network with $3$ convolutional layers and $1$ fully connected layer. The total number of trainable parameters for MNIST is $9,098$ and for CIFAR-10 is $21,840$. We simulate an FedMod system with $30$ UDs (for CIFAR-10) and  $20$ UDs (for MNIST) and $5$ UAVs each with $7$ orthogonal RRBs. In our experiments, we consider a network topology that is illustrated in Fig. \ref{fig3} unless otherwise
specified.  The remaining simulation parameters are summarized in TABLE \ref{table_1} and selected based on \cite{FL_Lit_1_4, zoheb, Mahdi, TWC, 1new}. To showcase the effectiveness of FedMoD in terms of learning accuracy and energy consumption, we consider the \textit{Star-based FL} and  \textit{HFL} schemes.

\begin{table}[t!] \small
	\renewcommand{\arraystretch}{0.9}
	\caption{Simulation Parameters}
	\label{table_1}
	\centering
	\begin{tabular}{|p{5.2cm}| p{3.0cm}|}
		\hline
		
		\textbf{Parameter} & \textbf{Value}\\
		\hline 
		Carrier frequency, $f$ & $1$ GHz \cite{zoheb}\\
		\hline
		Speed of light, $c$  & $3\times 10^8$ m/s\\
		\hline
		Propagation parameters, $a$ and $b$ & $9.6$ and $0.28$ \cite{zoheb}\\
		\hline
		Attenuation factors, $\psi^{LoS}$ and $\psi^{nLoS}$ & $1$ dB and $20$ dB \cite{zoheb} \\
		\hline 
		UAV's and UD's transmit maximum powers, $P$, $p$ & $1$ and $3$ Watt \cite{zoheb}\\
		\hline
		Transmit power of the CPS & $5$ Watt\\
		\hline
		Noise PSD, $N_0$ & -$174$ dBm/Hz\\
		\hline 
		Local and aggregated parameters size, $s$ & $9.1$ KB\\
		\hline
		
		UD processing density, $C_u$ & $[400-600]$\\
		\hline
	   UD computation frequency, $f_{u}$ & $[0.0003-1]$ G cycles/s\\
		\hline
		CPU architecture based parameter, $\alpha$ & $10^{-28}$\\
		\hline
		FL time threshold $T_{max}$ & $1$ Second \\
		\hline
		Number of data samples, $S_u$ & $200$\\
		\hline
	\end{tabular}
	
\end{table}

We show the training accuracy with
respect to number of iterations for both the MNIST and CIFAR-10 datasets with different model dissemination rounds $\alpha$ in Fig. \ref{fig4}. Specially, in Figs. \ref{fig4}(a) and \ref{fig4}(b), we show the accuracy performance of our proposed FedMoD scheme with full dissemination against the centralized FL schemes.  Particularly, in the considered star-based and HFL schemes, the CPS can receive  the local trained models from the UDs, where each scheduled UD transmits its trained model directly to the CPS (in case of star-based FL) or through UAVs (in case of HFL). Thus, the CPS can aggregate all the local models of all the scheduled UDs. In the considered decentralized FedMoD, before the dissemination rounds starts, each UAV has aggregated the trained local models of the scheduled UDs in its cluster only. However, with the novel dissemination FedMoD method, each UAV shares its aggregated models with the neighboring UAVs using one hop transmission. Thus, at each dissemination round, UAVs build their side information (\textit{Known} models and \textit{Unknown} models) until they receive all the  \textit{Unknown} models.  Thus, the UAVs have full knowledge about the global model of the system at each global iteration. Thanks to the efficient FedMoD dissemination method, the accuracy of the proposed FedMoD scheme is almost the same as the centralized FL schemes. Such efficient communications among UAVs accelerate the learning progress, thereby FedMoD model reaches an accuracy of ($0.945$, $0.668$ for MNIST and CIFAR-10) with around $200$ and $300$ global iterations, respectively, as compared to the accuracy of ($0.955$, $0.944$ for MNIST) and ($0.665$, $0.668$ for CIFAR-10) for star-based and HFL schemes, respectively. It is important to note that although our proposed overcomes the struggling UD issue of the star-based scheme and the two-hop transmission of the HFL, it needs a few rounds of model dissemination. However, the effective coding scheme of the models minimizes the number of dissemination rounds. In addition, due to the high communication links between the UAVs, the dissemination delay is negligible which does not affect the FL time.

\begin{figure}[t!]
	\centering
	\begin{subfigure}[t]{0.23\textwidth}
		\centerline{\includegraphics[width=1.13\linewidth]{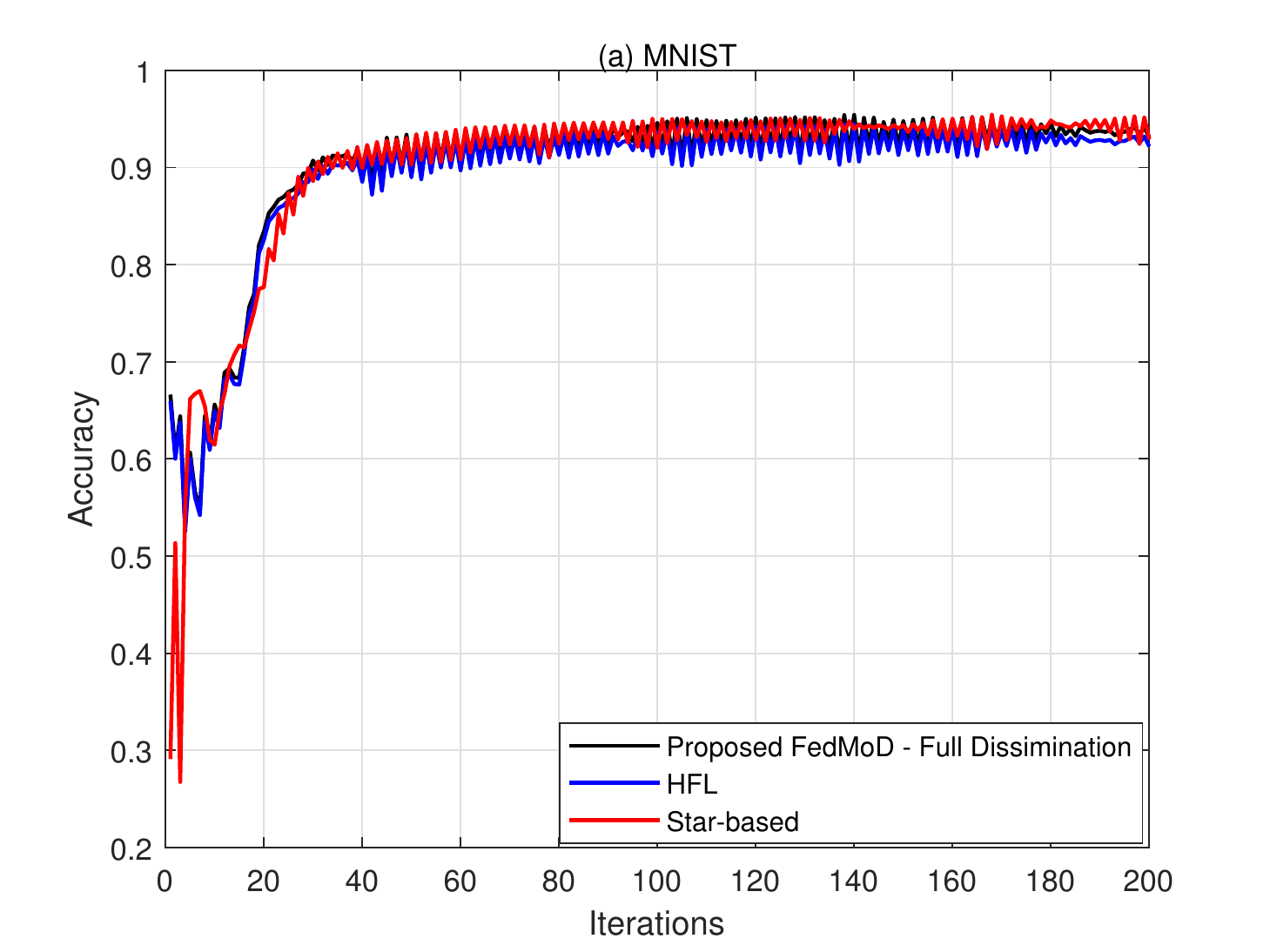}}
		\label{MCT}
	\end{subfigure}%
	~
	\begin{subfigure}[t]{0.23\textwidth}
		\centerline{\includegraphics[width=1.13\linewidth]{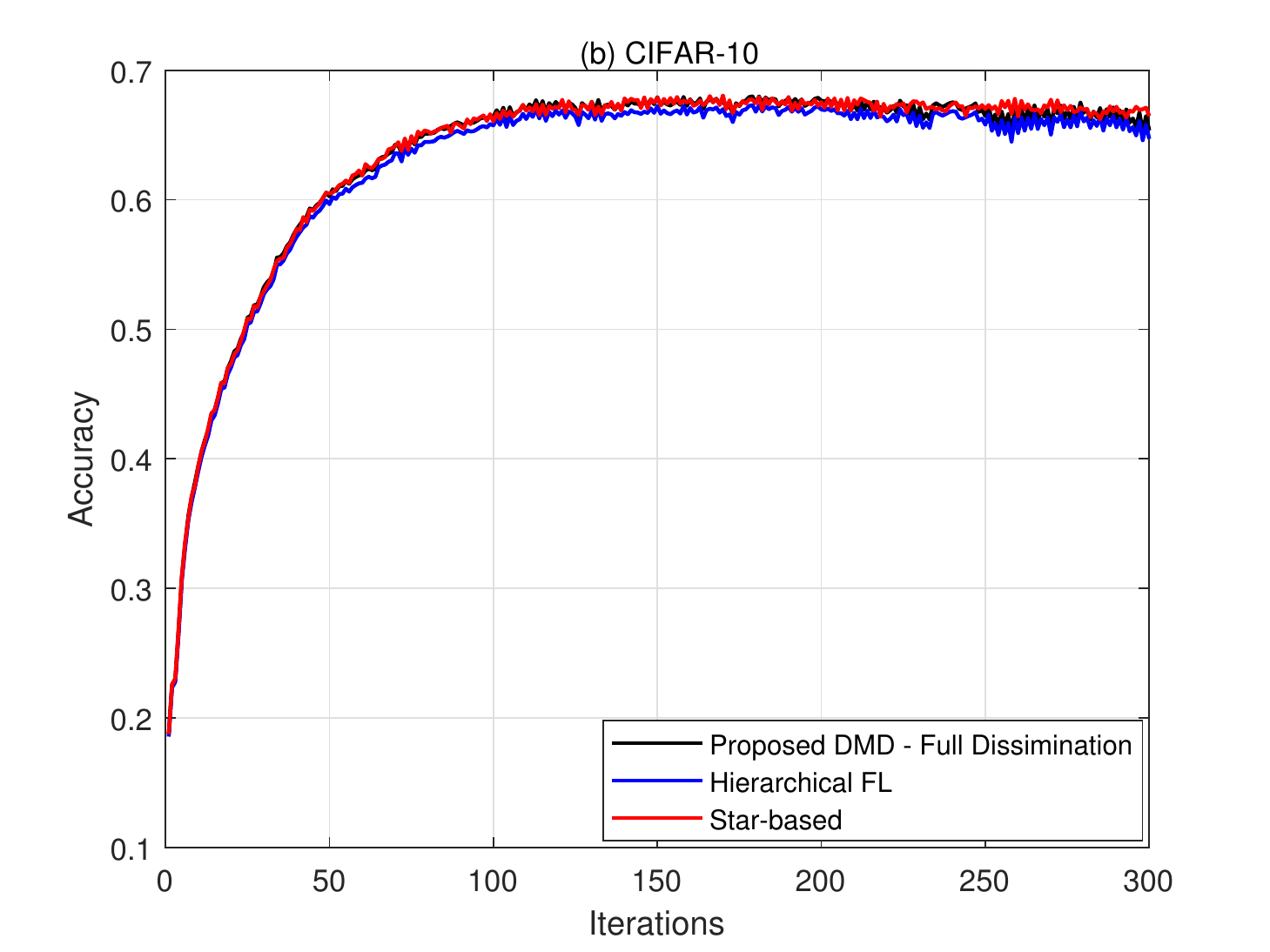}}
		\label{NCT}
	\end{subfigure}
	\caption{Performance comparison between FedMoD and baseline		schemes  for MNIST and CIFAR-10: Accuracy vs.  number of iterations.}
	\label{fig4}
\end{figure}

\begin{figure}[t!]
	\centering
	\begin{subfigure}[t]{0.23\textwidth}
		\centerline{\includegraphics[width=1.13\linewidth]{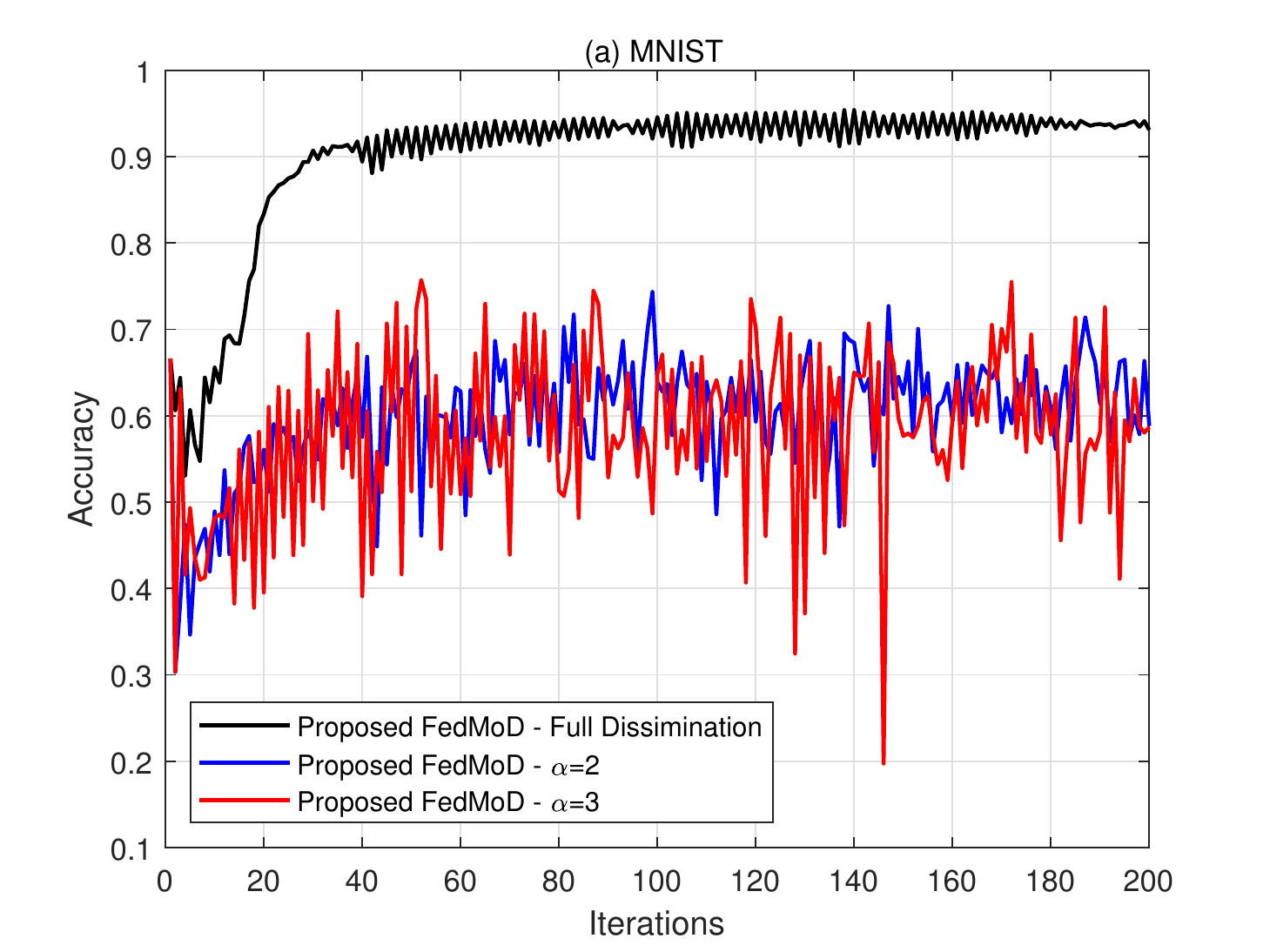}}
		\label{MCT}
	\end{subfigure}%
	~
	\begin{subfigure}[t]{0.23\textwidth}
		\centerline{\includegraphics[width=1.13\linewidth]{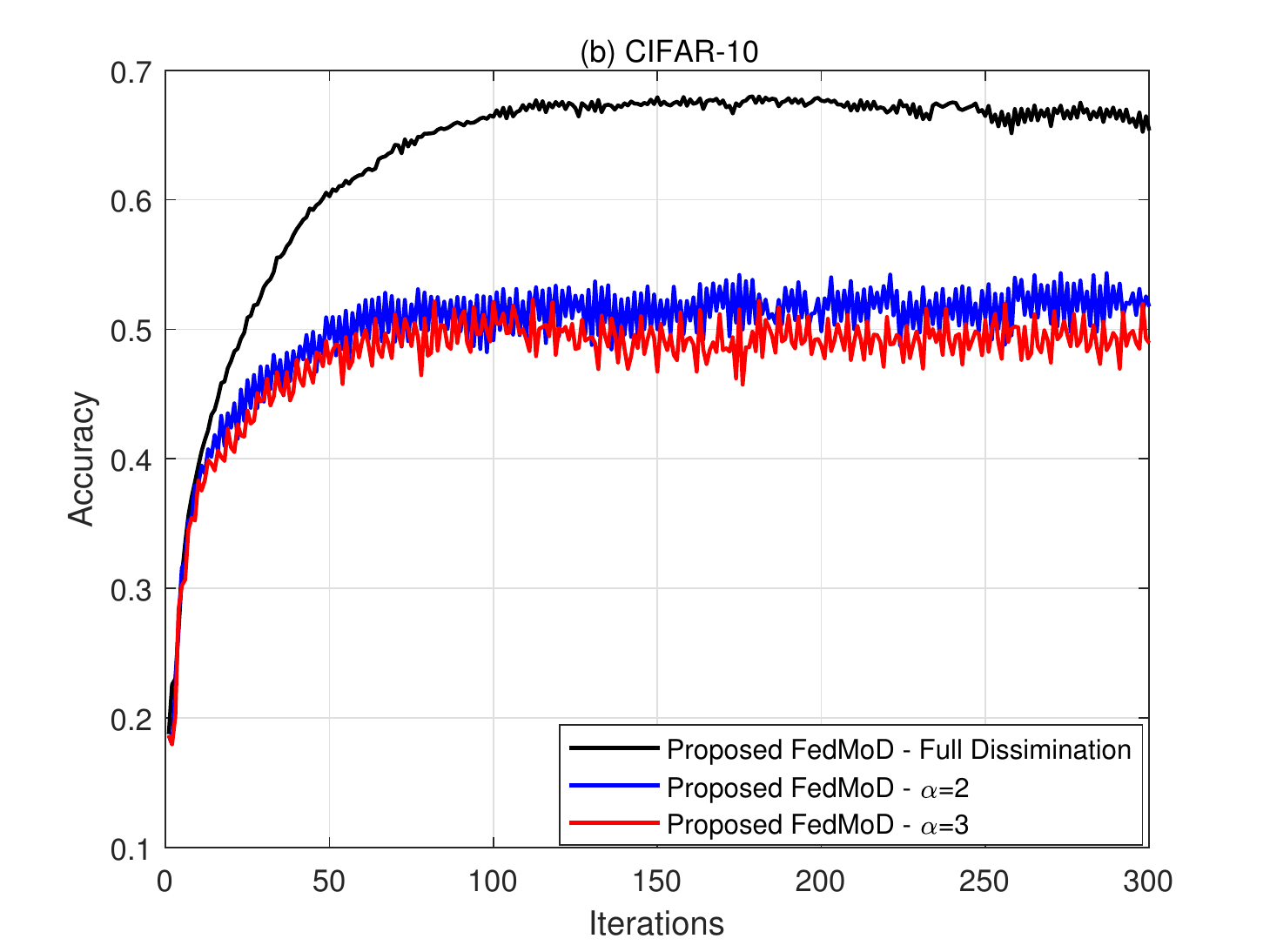}}
		\label{NCT}
	\end{subfigure}
	\caption{Performance comparison of FedMoD  for MNIST and CIFAR-10 with different $\alpha$.}
	\label{fig4a}
\end{figure}

In Figs. \ref{fig4a}(a) and \ref{fig4a}(b), we further study the impact of the number of dissemination rounds $\alpha$ on the convergence rate of the proposed FedMoD scheme for both the MNIST and CIFAR-10 datasets. For both figures, we consider the following three proposed schemes: (i) FedMoD scheme-full dissemination where UAVs perform full model dissemination at each global iteration,  (ii) FedMoD scheme - $\alpha=2$ where partially dissemination is performed and after each $2$ complete global iterations, we perform full dissemination, and (iii) FedMoD scheme - $\alpha=3$ where partially dissemination is performed and after each $3$ complete global iterations, a full dissemination is performed. From Figs. \ref{fig4a}(a) and \ref{fig4a}(b), we observe that a partial dissemination with less frequent full dissemination leads to a lower training accuracy within a given number of training iterations. Specifically, the accuracy performance for full dissemination, $\alpha=2$ and $3$ schemes is $0.966, 0.66, 0.75$ for MNIST and $0.668, 0.52, 0.59$ for CIFAR-10, respectively. Infrequent inter-cluster UAV dissemination also leads to un-stable convergence since the UAVs do not frequently aggregate all the local trained models of the UDs.

\begin{figure}[t!]
	\centering
	\begin{subfigure}[t]{0.23\textwidth}
		\centerline{\includegraphics[width=1.13\linewidth]{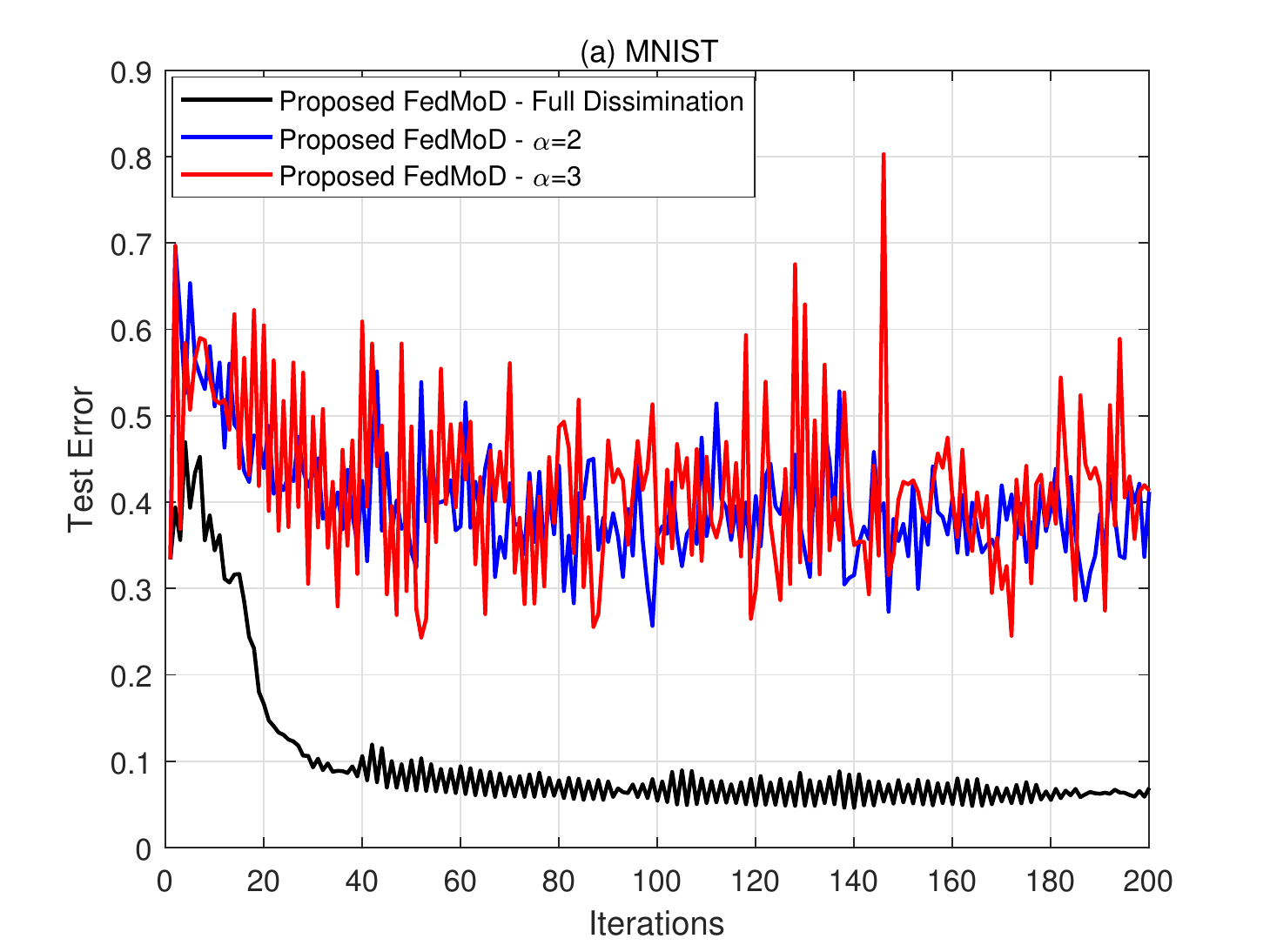}}
		\label{MCT}
	\end{subfigure}%
	~
	\begin{subfigure}[t]{0.23\textwidth}
		\centerline{\includegraphics[width=1.13\linewidth]{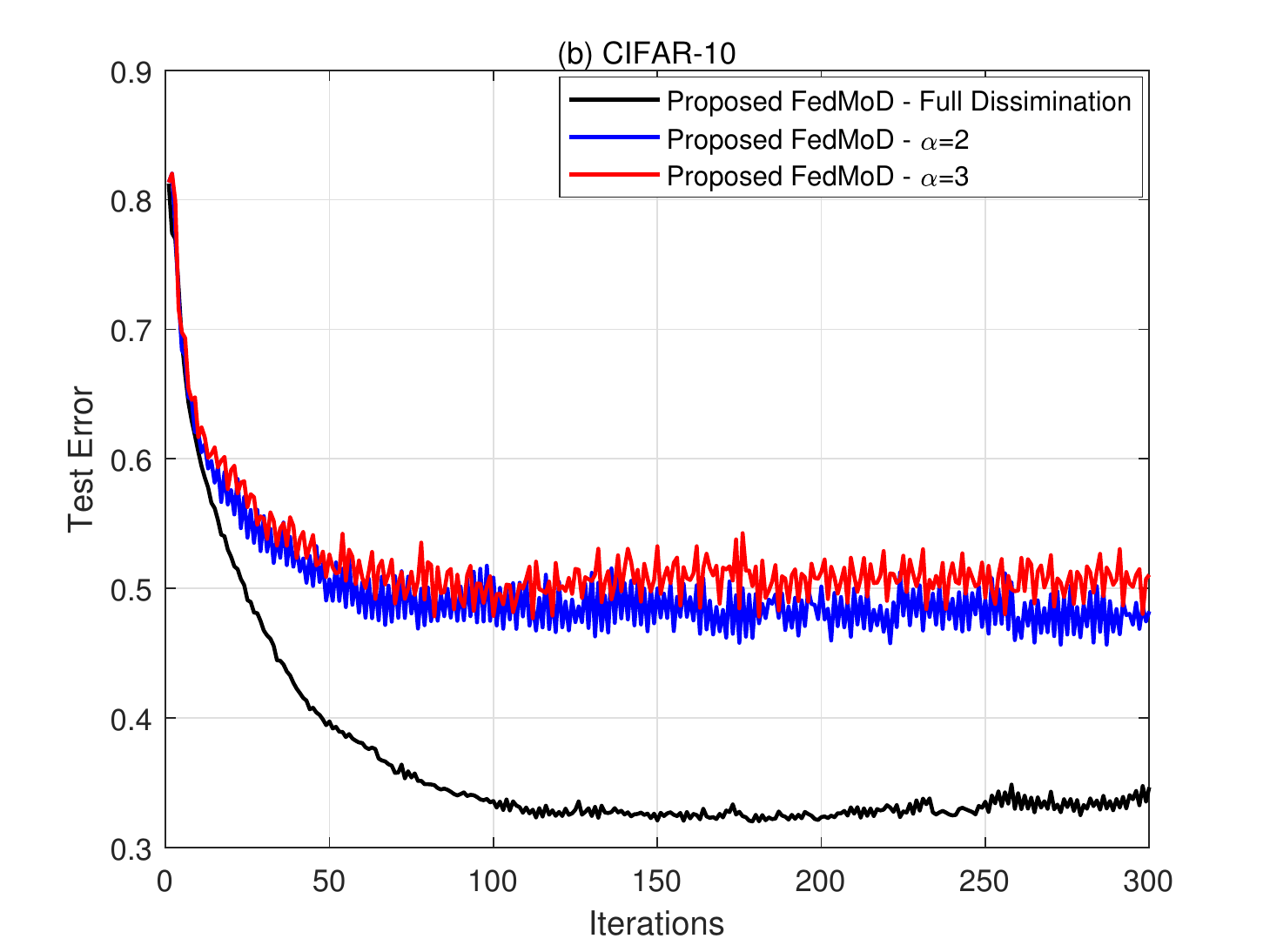}}
		\label{NCT}
	\end{subfigure}
	\caption{Test error of FedMoD  for MNIST and CIFAR-10 with different $\alpha$.}
	\label{fig5}
\end{figure}

In Figs. \ref{fig5}(a) and  \ref{fig5}(b),  we show the test error with respect to number of iterations for both the MNIST and CIFAR-10 datasets with different model dissemination rounds $\alpha$. It can be observed that, the test error of the proposed FedMod model drops rapidly at the early stage of the training process, and converges at around $160$ iterations. On the other hand, the training progress of FedMoD model with infrequent full model dissemination (i.e., $\alpha=2$, $\alpha=3$) lags far behind due to the insufficient model averaging of the UAVs, which is due to infrequent full communications among them. As result, both schemes do not converge and suffer higher testing loss compared to full dissemination of the FedMoD scheme since the communication among edge servers in the full dissemination is more efficient and thus accelerates the learning progress.

\begin{figure}[t!]
	\centering
	\begin{subfigure}[t]{0.23\textwidth}
		\centerline{\includegraphics[width=0.54\linewidth]{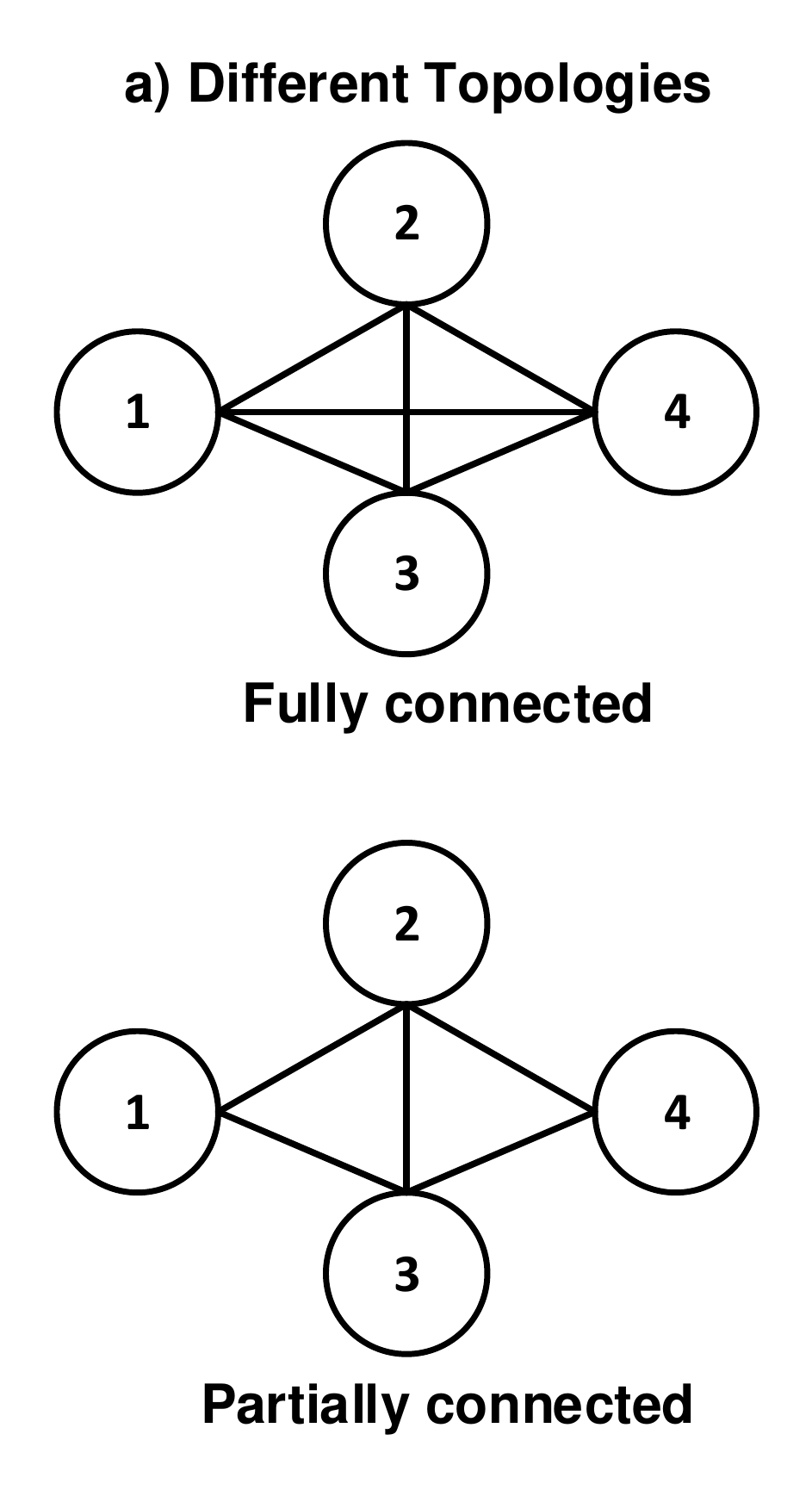}}
		\label{MCT}
	\end{subfigure}%
	~
	\begin{subfigure}[t]{0.23\textwidth}
		\centerline{\includegraphics[width=1.23\linewidth]{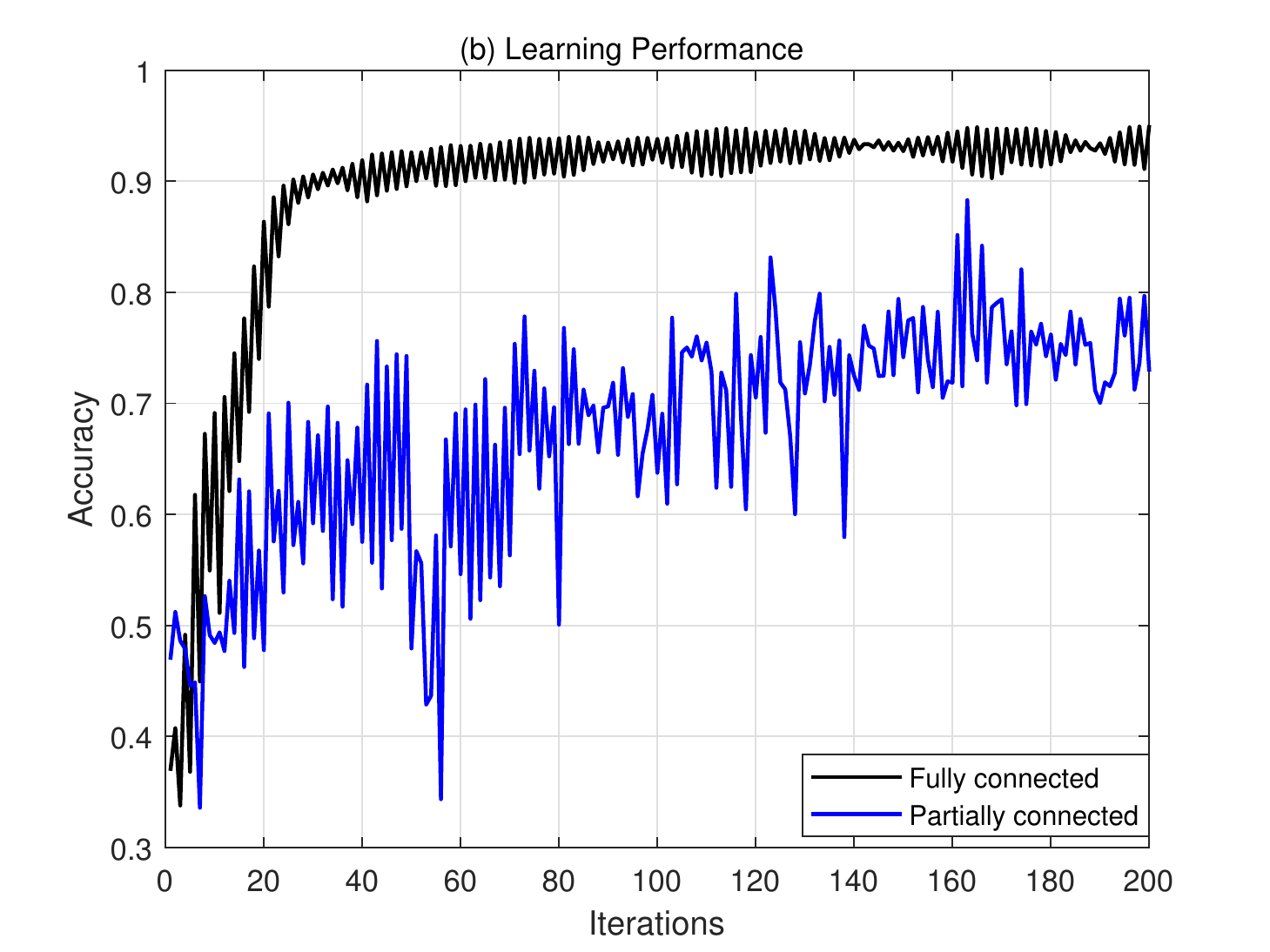}}
		\label{NCT}
	\end{subfigure}
	\caption{Typical network topologies of the UAVs for model dissemination and their FL accuracy.}
	\label{fig6}
\end{figure}

We also evaluate the learning accuracy of FedMoD on different
network topologies of the UAVs as shown in Fig. \ref{fig6}(a). We consider fully connected network where all the UAVs are connected and partially connected network where UAV $4$ is not connected to UAV $1$. In this figure, we perform $4$ different rounds of dissemination. As shown in Fig. \ref{fig6}(b),
we see that within a given number of global iterations, a more
connected network topology achieves a higher test accuracy. This is because more model information is collected from neighboring UAVs in each round of inter-UAV model aggregation. It is also observed that when $\alpha$ is greater
 than $4$, the test accuracy of the partially connected network can approach the case with a fully-connected
 network. Therefore, based on the network topology of
 UAVs, we can choose a suitable value of $\alpha$ to balance between the number of inter-cluster UAV aggregation and learning performance.

	\begin{figure}[t!]
		\centering
		
			\centerline{\includegraphics[width=0.55\linewidth]{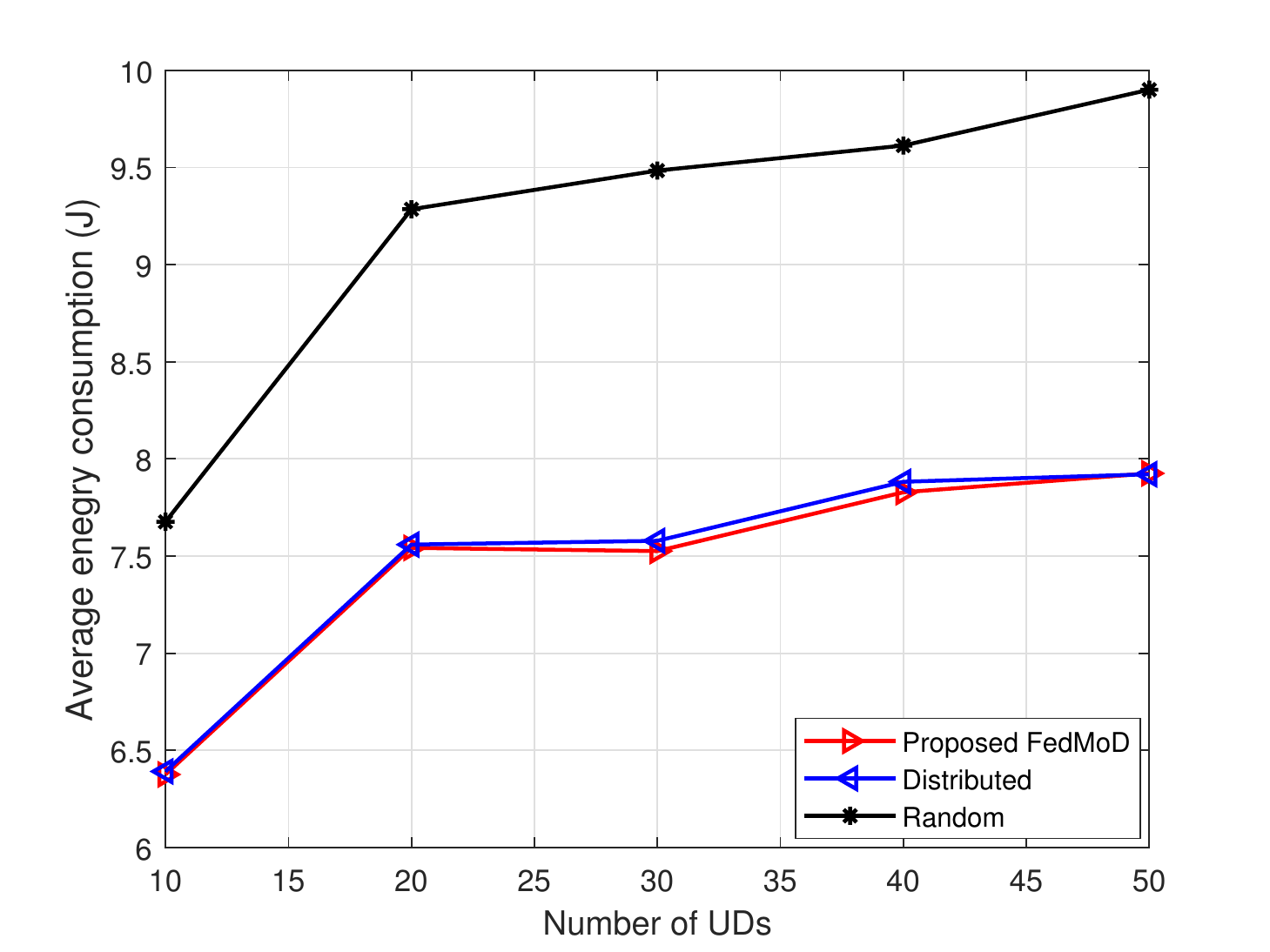}}
	
					\caption{Average energy consumption
vs. number of UDs.}
						\label{fig7}
\end{figure}

In Fig. \ref{fig7}, we plot the energy consumption of the proposed and benchmark schemes versus the number of UDs for a network of $4$ UAVs and $4$ RRBs per UAV. From the objective function of problem P1, we can observe that an efficient radio resource management scheme leads to
a lower energy consumption. Hence, from Fig.  \ref{fig7}, we observe that for FedMoD, the average energy consumption
is minimized. Such an observation is because the proposed schemes judiciously allocate LOS UDs to the UAVs and their available RRBs as well as D2D communications. In particular,  the random scheme has the largest energy consumption because it randomly schedules the UDs to the UAVs and their available RRBs. Accordingly, from
 energy consumption perspective, it is inefficient to consider a random radio resource management scheme. From Fig. \ref{fig7}, it is observed that the proposed centralized scheduling FedMoD and distributed FedMoD schemes offer the same energy consumption performances for the same number of UDs.  Such an observation can be explained by the following argument. When we have a large number of UDs, the probability that a UD is scheduled to more than one  UAV decreases. As a result,  the conflict among UAVs and the likelihood of scheduling UDs to the wrong UAV decreases.  As an insight from this figure,   a distributed radio resource management scheme  is a suitable alternative for scheduling the LOS UDs to the UAVs,  especially for large-scale networks.


\section{Conclusion}
In this paper, we developed a novel decentralized FL scheme, called FedMoD, which maintains convergence speed and reduces  energy consumption of FL in mmWave ATINs. Specifically, we proposed a FedMoD scheme based on inter-cluster UAV communications, and theoretically proved its convergence. A rate-adaptive and D2D-assisted RRM scheme was also developed to minimize the overall energy consumption of the proposed decentralized FL scheme. The presented simulation results revealed that our proposed FedMoD achieves the same accuracy as the baseline FL scheme while substantially reducing energy consumption for convergence. In addition, simulation results reveal various insights concerning how the topology of the network impacts the number of inter-cluster UAV aggregations required for the convergence of FedMoD.

	\numberwithin{equation}{section}
\setcounter{equation}{0}

\vspace{-0.33cm}
\appendices
\vspace{-0.33cm}

	\section{Proof of \lref{lemma1}}\label{A}
\vspace{-0.33cm}	
	By applying the global loss function to both sides of \eref{agg3} and using the L-smoothness assumption, we have:
	\begin{align} \nonumber
	&\mathbb E[F(\mathbf{\tilde{u}}(t+1))]
	\leq \mathbb E [F(\mathbf{\tilde{u}}(t))]+\mathbb E \bigg \langle \triangledown F(\mathbf{\tilde{u}}(t)),- \lambda {\mathbf G}(t) \mathbf m \bigg \rangle\\& \nonumber +\frac{L}{2} \mathbb E \lVert \lambda {\mathbf G}(t)\mathbf m \lVert^2_2
	=\mathbb E [F(\mathbf{\tilde{u}}(t))]-\delta \mathbb E \bigg \langle \triangledown F(\mathbf{\tilde{u}}(t)),\mathbb E [{\mathbf G}(t)\mathbf m] \bigg \rangle \\ & \nonumber+\frac{\lambda^2L}{2} \mathbb E \bigg \lVert    \mathbf G(t) \mathbf m  - \triangledown \tilde{\mathbf F}(t) \mathbf m+\triangledown \tilde{\mathbf F}(t)\mathbf m \bigg \lVert^2_2,
	\end{align}
	where $\triangledown \tilde{\mathbf F}(t)=\bigg[\triangledown F_1(t), \triangledown F_2(t), \cdots, \triangledown F_{U_{inv}}(t)\bigg]$. Since $\mathbb E [\hat{\mathbf L}(\mathbf{\tilde{w}}_{i}(t-y))\hat{\mathbf m}]=\triangledown \tilde{\mathbf L}(\mathbf{\tilde{w}}_{i}(t-y)) \hat{\mathbf m}$, we have the following
	\begin{align} \nonumber
	&\mathbb E[F(\mathbf{\tilde{u}}(t+1))]=\mathbb E [F(\mathbf{\tilde{u}}(t))]-\lambda \mathbb E \bigg \langle \triangledown F(\mathbf{\tilde{u}}(t)),\triangledown \tilde{\mathbf F}(t) \mathbf m \bigg \rangle \\& \nonumber + \frac{\lambda^2L}{2} \mathbb E \bigg \lVert  {\mathbf G}(t)\mathbf m - \triangledown \tilde{\mathbf F}(t)\mathbf m \bigg\lVert^2_2 + \frac{\lambda^2L}{2} \lVert \triangledown \tilde{\mathbf F}(t)\mathbf m  \lVert^2_2,
	\end{align}
	where $\frac{\lambda^2L}{2} \mathbb E \bigg \lVert  {\mathbf G}(t)\mathbf m - \triangledown \tilde{\mathbf F}(t)\mathbf m+\triangledown \tilde{\mathbf F}(t)\mathbf m \bigg \lVert^2_2=\frac{\lambda^2L}{2} \mathbb E \bigg \lVert   {\mathbf G}(t) \mathbf m - \triangledown \tilde{\mathbf F}(t) \mathbf m \bigg \lVert^2_2 + \frac{\lambda^2L}{2} \mathbb E  \lVert \triangledown \tilde{\mathbf F}(t) \mathbf m  \lVert^2_2$. This is becasue $\mathbb E [\mathbf G(t) \mathbf m]=\triangledown \tilde{\mathbf F}(t)  \mathbf m$, thus the cross-terms of $  \mathbf G(t) \mathbf m$ and $\triangledown \tilde{\mathbf F}(t) \mathbf m$ are zero. Thus, we have 
	\begin{align}  \nonumber
	&\mathbb E[F(\mathbf{\tilde{u}}(t+1))]=\mathbb E [F(\mathbf{\tilde{w}}(t))]- \delta \mathbb E \bigg \langle \triangledown F(\mathbf{\tilde{u}}(t)),\\& \nonumber \sum_{u=1}^{N_{inv}} m_u \triangledown F_u  (\mathbf{\tilde{u}}_{u}(t) ) \bigg \rangle + \frac{\lambda^2 L}{2} \mathbb E \bigg \lVert \sum_{u=1}^{U_{inv}}  m_u \bigg ( g (\mathbf{\tilde{u}}_{u}(t))\\& \nonumber - \triangledown F  (\mathbf{\tilde{w}}_{u}(t) )\bigg) \bigg \lVert^2_2 +\frac{\lambda^2L}{2} \lVert \triangledown \tilde{\mathbf F}(t) \mathbf m   \lVert^2_2 \\& \nonumber 
	= \mathbb E [F(\mathbf{\tilde{u}}(t))]- \lambda \mathbb E \bigg \langle \triangledown F(\mathbf{\tilde{u}}(t)),\sum_{u=1}^{U_{inv}} m_u \triangledown F_i (\mathbf{\tilde{w}}_{u}(t)) \bigg \rangle \\& \nonumber +\frac{\lambda^2 L}{2} \sum_{u=1}^{U_{inv}} m^2_u \mathbb E \bigg \lVert   g (\mathbf{\tilde{w}}_{u}(t) )- \triangledown F_u  (\mathbf{\tilde{w}}_{u}(t) )\bigg \lVert^2_2+\frac{\lambda^2L}{2} \lVert \triangledown \tilde{\mathbf F}(t) \mathbf m  \lVert^2_2.
	\end{align}	
 
	From Assumption 3,  we have $\mathbb E \bigg \lVert   g (\mathbf{\tilde{w}}_{u}(t) )- \triangledown F_u  (\mathbf{\tilde{w}}_{u}(t) )\bigg \lVert^2_2 \leq \sigma^2$. Also, $\mathbf s^T \mathbf k= \frac{1}{2} \lVert  \mathbf s\lVert^2_2 + \frac{1}{2} \lVert  \mathbf k\lVert^2_2 -\frac{1}{2} \lVert  \mathbf s- \mathbf k \lVert^2_2$. Thus, we have $ \lambda \mathbb E \bigg \langle \triangledown F(\mathbf{\tilde{u}}(t)),\sum_{u=1}^{U_{inv}} m_u \triangledown F_u  (\mathbf{\tilde{w}}_{u}(t)) \bigg \rangle = \frac{\lambda}{2} \lVert  \triangledown F(\mathbf{\tilde{u}}(t)) \lVert^2_2 \\+\frac{\lambda}{2} \mathbb E  \lVert \sum_{u=1}^{U_{inv}} m_u \triangledown F  (\mathbf{\tilde{w}}_{u}(t) ) \lVert^2_2- \frac{\lambda}{2} \mathbb E \bigg \lVert  \triangledown F(\mathbf{\tilde{u}}(t))- \sum_{u=1}^{U_{inv}} m_u \triangledown F  (\mathbf{\tilde{w}}_{u}(t) ) \bigg \lVert^2_2$. With such a constraint, we can have   			
	\begin{align} \nonumber
	&\mathbb E[F(\mathbf{\tilde{u}}(t+1))] \leq \mathbb E [F(\mathbf{\tilde{u}}(t))]- \bigg (\frac{\lambda}{2} \lVert  \triangledown F(\mathbf{\tilde{u}}(t)) \lVert^2_2 \\& \nonumber+\frac{\lambda}{2} \mathbb E \bigg \lVert \sum_{u=1}^{U_{inv}} m_u \triangledown F  (\mathbf{\tilde{u}}_{u}(t) ) \bigg \lVert^2_2  -\frac{\lambda}{2} \mathbb E \bigg \lVert  \triangledown F(\mathbf{\tilde{u}}(t))  \\& \nonumber - \sum_{u=1}^{U_{inv}} m_u \triangledown F  (\mathbf{\tilde{w}}_{i}(t) ) \bigg \lVert^2_2 \bigg )+ \frac{\lambda^2 L}{2}\sum_{u=1}^{U_{inv}} m^2_{u} \sigma^2+\frac{\lambda^2L}{2}  \lVert \triangledown \tilde{\mathbf F}(t)\mathbf m   \lVert^2_2 \\& \nonumber 
	\leq \mathbb E [F(\mathbf{\tilde{u}}(t))]- \bigg (\frac{\lambda}{2} \lVert  \triangledown F(\mathbf{\tilde{u}}(t)) \lVert^2_2 +\frac{\lambda}{2} \mathbb E \bigg \lVert \sum_{u=1}^{U_{inv}} m_u \triangledown F  (\mathbf{\tilde{w}}_{u}(t)) \bigg \lVert^2_2 \\& \nonumber -\frac{\lambda}{2} \mathbb E \bigg \lVert  \sum_{u=1}^{U_{inv}} m_u \bigg (\triangledown F(\mathbf{\tilde{u}}(t))  -   \triangledown F  (\mathbf{\tilde{w}}_{u}(t))  \bigg ) \bigg \lVert^2_2 \bigg )\\ & \nonumber +\frac{\lambda^2 L}{2}\sum_{u=1}^{U_{inv}} m^2_{u} \sigma^2+\frac{\lambda^2L}{2} \mathbb E \lVert \triangledown \sum_{u=1}^{U_{inv}} m_u \triangledown F_u(\mathbf{\tilde{w}}_u(t)) \lVert^2_2.
	\end{align}

	We denote $\tilde{Q}=\mathbb E \lVert \triangledown \sum_{i=1}^{U_{inv}} m_u \triangledown F_u(\mathbf{\tilde{w}}_u(t)) \lVert^2_2$,  we have the following 
	\begin{align} \nonumber
	&\mathbb E[F(\mathbf{\tilde{u}}(t+1))] \leq  \mathbb E [F(\mathbf{\tilde{u}}(t))]- \frac{\lambda}{2} \lVert  \triangledown F(\mathbf{\tilde{u}}(t)) \lVert^2_2 -\bigg ( \frac{\lambda}{2} -\frac{\lambda^2 L}{2}\bigg ) \tilde{Q} \\& \nonumber +\frac{\lambda}{2} \mathbb E \bigg \lVert  \sum_{u=1}^{U_{inv}} m_u \bigg (\triangledown F(\mathbf{\tilde{u}}(t))  -   \triangledown F  (\mathbf{\tilde{w}}_u(t))  \bigg ) \bigg \lVert^2_2+\frac{\lambda^2 L}{2}\sum_{u=1}^{U_{inv}} m^2_{u} \sigma^2 \\& \nonumber 
	=  \mathbb E [F(\mathbf{\tilde{u}}(t))]- \frac{\lambda}{2} \lVert  \triangledown F(\mathbf{\tilde{u}}(t)) \lVert^2_2 -\bigg ( \frac{\lambda}{2} -\frac{\lambda^2 L}{2}\bigg ) \tilde{Q} \\& \nonumber +\frac{\lambda}{2} \mathbb E \bigg \lVert  \sum_{u=1}^{U_{inv}} m_u \mathbb E \bigg \lVert \triangledown F(\mathbf{\tilde{u}}(t))  -   \triangledown F  (\mathbf{\tilde{w}}_i(t))    \bigg \lVert+\frac{\lambda^2 L}{2}\sum_{u=1}^{U_{inv}} m^2_{u} \sigma^2
	\\& \nonumber 
	\leq  \mathbb E [F(\mathbf{\tilde{u}}(t))]- \frac{\lambda}{2} \lVert  \triangledown F(\mathbf{\tilde{u}}(t)) \lVert^2_2 -\bigg ( \frac{\lambda}{2} -\frac{\lambda^2 L}{2}\bigg ) \tilde{Q} \\& \nonumber +\frac{\lambda}{2} \mathbb E \bigg \lVert  \sum_{u=1}^{U_{inv}} m_u \mathbb E \bigg \lVert  \mathbf{\tilde{u}}(t)-\mathbf{\tilde{w}}_u(t)  \bigg \lVert+\frac{\lambda^2 L}{2}\sum_{u=1}^{U_{inv}} m^2_{u} \sigma^2.
	\end{align}	
	The last inequality holds because of the L-
	smoothness assumption of the local loss function.	We conclude the proof by moving $\mathbb E [F(\mathbf{\tilde{u}}(t))]$ to the left hand side (LHS), thus we will have
	\begin{align} \nonumber
	&\mathbb{E}[F(\mathbf{\tilde{u}}(t+1))]-\mathbb{E}[F(\mathbf{\tilde{u}}(t))]  \nonumber \leq \frac{-\lambda}{2} \mathbb E  \lVert \triangledown F(\mathbf{\tilde{u}}(t)) \lVert^2_2\\ &+\frac{\lambda^2 L}{2} \sum_{u=1}^{{U_{inv}}} m_u\sigma^2  -\frac{\lambda}{2}(1-\lambda L)\tilde{Q}+\frac{\lambda L^2}{2}\mathbb E \bigg \lVert \mathbf{\tilde{w}}(t)(\mathbf I-\mathbf M) \bigg \lVert^2_\mathbf M.
	\end{align}

\end{document}